\documentclass[11pt]{amsart}
\usepackage{amsmath,amsfonts,amsthm,amssymb,bbm,tikz,url}
\usetikzlibrary{decorations.pathmorphing}
\usepackage[margin=1.in]{geometry}

\newcommand{\reals}{\mathbb{R}}
\newcommand{\comps}{\mathbb{C}}

\newcommand{\nats}{\mathbb{N}}

\newcommand{\Hi}{\mathcal{H}}
\newcommand{\id}{\mathbbm{1}}
\newcommand{\C}{\mathbb{C}}

\newcommand{\N}{\mathbb{N}}

\newcommand{\Hil}{\mathcal{H}}

\newcommand{\unity}{\mathbbm{1}}

\newtheorem{thm}{Theorem}[section]
\newtheorem{remark}[thm]{Remark}
\newtheorem{lemma}[thm]{Lemma}
\newtheorem{prop}[thm]{Proposition}
\newtheorem{cor}[thm]{Corollary}
\newtheorem{conj}[thm]{Conjecture}

\theoremstyle{definition}
\newtheorem{example}[thm]{Example}

\DeclareMathOperator{\spec}{spec}

\begin{document}

\renewcommand{\thefootnote}{\fnsymbol{footnote}}
\title[$U_q(\mathfrak{sl}_2)$ invariant spin chains]{Ferromagnetic Ordering of Energy Levels\\
for  $U_q(\mathfrak{sl}_2)$ Symmetric Spin Chains}

\author[B. Nachtergaele]{Bruno Nachtergaele}
\address{Department of Mathematics\\
University of California, Davis\\
Davis, CA 95616, USA}
\email{bxn@math.ucdavis.edu}

\author[S. Ng]{Stephen Ng}
\address{Department of Mathematics\\
University of California, Davis\\
Davis, CA 95616, USA}
\email{stephenng@math.ucdavis.edu}

\author[S. Starr]{Shannon Starr}
\address{Department of Mathematics\\
University of Rochester\\
Rochester, NY 14627, USA}
\email{starr@math.rochester.edu}

\date{Version: \today }
\maketitle
\bigskip
\begin{abstract}
We consider the class of quantum spin chains with arbitrary 
$U_q(\mathfrak{sl}_2)$-invariant nearest neighbor interactions,
sometimes called $\textrm{SU}_q(2)$ for the quantum deformation
of $\textrm{SU}(2)$, for $q>0$.
We derive sufficient conditions for the Hamiltonian to satisfy the property
we call {\em Ferromagnetic Ordering of Energy Levels}.
This is the property that the ground state energy restricted to a fixed total spin
subspace is a decreasing function of the total spin.
Using the Perron-Frobenius theorem, we show sufficient conditions
are positivity of all interactions in the dual canonical basis of Lusztig.
We characterize the cone of positive interactions, showing that
it is a simplicial cone consisting of all non-positive linear combinations
of ``cascade operators,'' a special new basis of $U_q(\mathfrak{sl}_2)$
intertwiners we define.
\end{abstract}

\maketitle

\footnotetext[1]{Copyright \copyright\ 2011 by the authors. This
paper may be reproduced, in its entirety, for non-commercial
purposes.}

\section{Introduction}\label{sec:intro}

\newcommand{\U}{U_q(\mathfrak{sl}_2)}
\newcommand{\B}[1]{B_{#1}^{TL}}
\newcommand{\qbinom}[2]{\begin{bmatrix} #1 \\ #2 \end{bmatrix}}
\newcommand{\up}{\: \raisebox{-.25cm}{\begin{tikzpicture} \draw[thick,->] (0,0) -- (0,.5);\end{tikzpicture}} \:}
\newcommand{\down}{\: \raisebox{-.25cm}{\begin{tikzpicture} \draw[thick,<-] (0,0) -- (0,.5);\end{tikzpicture}} \:}
\newcommand{\sector}[2]{\mathcal{N}(#1,#2)}

Symmetry considerations often play a central role in the study of quantum
spin chains. This is true for models that allow for exact solutions, using the Bethe 
Ansatz \cite{korepin:2006} or by other methods \cite{affleck:1988,kirillov:1989}, as well
for the analysis of 
general properties such as the uniqueness or degeneracy of the ground state and
the nature of the excited states above it \cite{lieb:1962,affleck:1986}. In fact, it is precisely
the mathematical study of quantum spin chains and their symmetries that led to the 
notion of quantum group and other important developments in representation theory. 

Apart from symmetry, some quantum spin Hamiltonians also possess an interesting 
positivity structure which allows them, and the dynamics they generate, to be related
to stochastic dynamics of classical probabilistic models. As a prominent example,
we mention the XXZ chain.The generator of the simple exclusion process for particles 
on a graph is related to the spin-1/2 XXX model on that graph by a similarity transformation. 
As a consequence, questions about mixing rates of the simple exclusions process can be 
related to properties of the spectrum of the XXX Hamiltonian. E.g., it has been known for 
some time that the Aldous Conjecture \cite{AldousConjecture}, i.e., the property that the 
relaxation time is independent  of the number of particles, is implied by Ferromagnetic Ordering 
of Energy Levels (FOEL) in the XXX model \cite{nachtergaele:2006}.
The FOEL property states that for certain graphs, if one considers $E_0(s)$ to be the minimum
energy eigenvalue of the XXX Hamiltonian among all eigenvectors with total spin equal to $s$,
then $E_0(s)$ is an decreasing function of $s$. 

The Aldous Conjecture for 
general graphs was recently proved in  \cite{CaputoLiggettRichthammer}.
In one dimension, one may relate the asymmetric exclusion process to the XXZ model.
A generalization of Aldous's conjecture in this context follows from \cite{KomaNachtergaele},
and FOEL was proved in \cite{nachtergaele:2004} using the quantum group $\U$  to define spin.
In this paper
we investigate the FOEL property for a more general class of models.

A surprising example where symmetry and positivity conspire to relate quantum
spin chains with {\em a priori} unrelated mathematical problems is given by the 
Razumov-Stroganov Conjecture, which was also proved recently \cite{CantiniSportiello}.
In this case, the expansion coefficients of the ground state of the XXZ chain with a specific 
value of the anisotropy parameter, and properly normalized, count the number of
Alternating Sign Matrices (ASM) restricted by a boundary condition. One can see this as 
a refinement of the previously established connection between the six-vertex model and 
the combinatorics of ASM \cite{kuperberg:1996,kuperberg:2002}.

All these connections obviously require that the quantum
spin Hamiltonian, in the case of applications in probability, or the eigenvector,
in the case of a combinatorial interpretation, need to be given by non-negative numbers.
In this work we study the class of $\U$-invariant 
quantum spin Hamiltonians with nearest neighbor interactions that have non-negative
matrix elements in a suitable basis and we show that this is a sufficient condition
for Ferromagnetic Ordering of Energy Levels. FOEL is the property that
the smallest eigenvalues of the restrictions of the Hamiltonian to the invariant
subspaces given by the eigenspaces of the $\U$ Casimir operator, are 
a non-increasing function of the Casimir eigenvalue. This is the hallmark of
a ferromagnetic quantum spin chain. FOEL immediately implies that the ground state
space coincides with the irreducible representation of $\U$ of the maximal dimension
occurring in the spin chain and that the first excited state belongs to an irreducible representation
of the the next largest dimension, which are known as spin waves in the physics
literature. 

The FOEL property has been proved in previous works for
the spin-1/2 XXZ Heisenberg chain \cite{nachtergaele:2004}, and the XXX Heisenberg
ferromagnetic chains of arbitrary spins \cite{nachtergaele:2005}. We note that for the XXX
antiferromagnetic chains (and for Heisenberg models on general bipartite
lattices) monotone behavior of the minimum energy states in
subspaces of fixed total spin was proved by Lieb and Mattis \cite{lieb:1962}. 
Since the Heisenberg ferro- and antiferromagnetic Hamiltonians are not
unitarily equivalent, FOEL and the ordering proved by Lieb and Mattis are independent
properties.
The present work is a
generalization of FOEL to general $\U$-invariant nearest neighbor interactions.
Generalizations in other directions (Hubbard type and $SU(n)$-invariant models, and
ladder systems) have been considered by Hakobyan \cite{hakobyan:2004,hakobyan:2008,
hakobyanLadder:2010,hakobyan:2010}. 

In this work we use the graphical calculus of \cite{kauffman:1994}
and \cite{frenkel:1997} to derive new formulas to show that the Hamiltonians of a class 
of quantum spin chains, when expressed in the dual canonical basis, have all non-positive 
off-diagonal matrix elements. In \cite{nachtergaele:2004,nachtergaele:2005,nachtergaele:2006}, 
the proof of FOEL  proceeded by induction in the length of the chain. 
Here we provide an alternate proof which does not require induction. We retain the use of Perron-
Frobenius style arguments of the previous results but combine them with the $\U$ symmetry 
in a new way to prove FOEL for the class of $\U$-invariant Hamiltonians we consider.

In order to make this paper accessible to readers with a variety of backgrounds, we
have chosen to present the material in a self-contained manner. The main result is 
stated and illustrated in Section \ref{sec:statement}. Section \ref{sec:prob} is a brief 
discussion of the implications of FOEL for the probabilistic models associated with the
quantum spin chains we study.
In Section \ref{sec:graphical_calculus}, 
we give a concise review of the graphical calculus for $TL_n$, the Temperley-Lieb algebra,
which is the algebra of intertwiners $\textbf{End}_{\U}(V(1)^{\otimes n})$. We also introduce
various bases for the space of intertwiners, including what we call the ``cascade basis.'' Next, in Section \ref{sec:dual_canonical_basis}, 
we introduce the Dual Canonical Basis and 
calculate the matrix entries of the various intertwiners in this basis.
The Perron-Frobenius type argument is explained in Section \ref{sec:sufficient_conditions}.
In that section we also give the proof of the main result and show that the 
cone spanned by non-positive combinations of the ``cascade basis'' is the maximal
set of interactions which have positive matrix entries in the dual canonical basis.

\section{Statement of Results for $\text {SU}(2)$ Models}
\label{sec:statement}

In this section, we will state our main results focusing first on the case of the classical Lie group
$\text {SU}(2)$.

Consider a quantum spin chain of length $L$,
and suppose that for each $x \in \{1,\dots,L\}$,
the spin at site $x$ has total spin $s_x$.
Therefore, there is a single-site Hilbert space $\Hil_x \cong \C^{2s_x+1}$
with spin operators $S_x^{(1)},S_x^{(2)},S_x^{(3)}$ satisfying the $\text {SU}(2)$ commutation relations and
$$
\boldsymbol{S}_x^2\, :=\, [S_x^{(1)}]^2 + [S_x^{(2)}]^2 + [S_x^{(3)}]^2\, =\, s_x (s_x + 1) \unity_x\, .
$$
The full Hilbert space is 
$$
\Hil_{[1,L]}\, =\, \Hil_1 \otimes \cdots \otimes \Hil_L\, .
$$
We identify the operator $S_x^{(a)}$ with
$$
\unity_1 \otimes \cdots \otimes \unity_{x-1} \otimes S_x^{(a)} \otimes \unity_{x+1} \otimes \cdots \otimes \unity_L\, .
$$
The general $\text {SU}(2)$-symmetric,
 nearest-neighbor interaction between spins at sites $x$ and $x+1$ is
$$
h_{x,x+1}\, =\, \alpha_{x,x+1}^{(0)} \unity + \sum_{k=1}^{n_{x,x+1}} \alpha_{x,x+1}^{(k)} (\boldsymbol{S}_x \cdot \boldsymbol{S}_{x+1})^k\, ,
$$
for some constants $\alpha_{x,x+1}^{(0)},\dots,\alpha_{x,x+1}^{(n_{x,x+1})}$,
where $n_{x,x+1} = 2(s_x \wedge s_{x+1})$,
and 
$$
\boldsymbol{S}_x \cdot \boldsymbol{S}_{x+1}\, =\, S_x^{(1)} S_{x+1}^{(1)} + S_x^{(2)} S_{x+1}^{(2)} + S_x^{(3)} S_{x+1}^{(3)}\, .
$$
The full Hamiltonian on $\Hil_{[1,L]}$ is 
$$
H_{[1,L]}\, =\, \sum_{x=1}^{L-1} h_{x,x+1}\, .
$$
An important problem is to diagonalize such quantum spin chains,
or to determine qualitative features of the spectrum based on properties
of the interactions $h_{x,x+1}$.

The Hamiltonian has the full $\textrm{SU}(2)$ symmetry.
In fact, for any $H$ that commutes with $\textrm{SU}(2)$, 
we let
$$
E_0(H,s)\, =\, \inf \spec H_s 
$$
where $H_s$ is the $H$ restricted to the space of total spin $s$. 
We say that $H$ has the ferromagnetic ordering of energy levels property, or FOEL for short, if
$$
s\leq s' \ \Rightarrow \ E_0(H,s) \geq E_0(H,s')\, .
$$
In this paper we consider conditions on the nearest neighbor interactions $h_{x,x+1}$ guaranteeing that 
$H_{[1,L]}$ satisfies the FOEL property.

\subsection{The case of two spins}

It is instructive to first consider the conditions for the FOEL property for the simplest case, where $L=2$ so that there are just two sites.
We consider the magnitudes of the spins to be $s_1,s_2 \in \{\frac{1}{2},1\frac{3}{2},2,\dots\}$.
The spectrum of
$$
|\boldsymbol{S}_1|^2 + 2 \boldsymbol{S}_1 \cdot \boldsymbol{S}_2 + |\boldsymbol{S}_2|^2\, ,
$$
is $\{j(j+1)\, :\, j \in \mathcal{J}(s_1,s_2)\}$,
where
$$
\mathcal{J}(s_1,s_2)\, =\, \{|s_1-s_2|,|s_1-s_2|+1,\dots,s_1+s_2\}\, .
$$
Note that $|\mathcal{J}(s_1,s_2)| = n_{1,2}+1$, where $n_{1,2} = 2(s_1\wedge s_2)$ was defined above.
We shift the basic Heisenberg interaction so as to have ground state equal to $0$:
$$
\mathfrak{h}_{s_1,s_2}\, =\, s_1 s_2 \mathbbm{1} - \boldsymbol{S}_1 \cdot \boldsymbol{S}_2\, .
$$
The spectrum of this operator is $\{E_{s_1,s_2}(j)\, :\, j \in \mathcal{J}(s_1,s_2)\}$, where
$$
E_{s_1,s_2}(j)\, =\, \frac{(s_1+s_2)(s_1+s_2+1) - j(j+1)}{2}\, .
$$
Note that there are only $n_{1,2}+1$ different eigenvalues.
Therefore any function of this finite set of eigenvalues may be represented as a polynomial
of degree at most $n_{1,2}$.
Therefore, the most general interaction we need to consider for a chain of length 2
is
$$
H_{[1,2]}\, =\, Q(\mathfrak{h}_{s_1,s_2})\, ,
$$
where $Q(z) = a_0 + a_1 z + \dots + a_n z^n$ is a general polynomial, with real coefficients, of degree at most $n_{1,2}$.
The general condition for $H_{[1,2]}$ to satisfy the FOEL property is that
$$
Q\big(E_{s_1,s_2}(j)\big)\, \leq\, Q\big(E_{s_1,s_2}(j')\big)\, ,
$$
whenever $j\geq j'$ for $j,j' \in \mathcal{J}(s_1,s_2)$.
The set of such polynomials forms a cone.
More precisely, let us define for each $j \in \mathcal{J}(s_1,s_2)$, the polynomial $Q_j(z)$ such that
\begin{equation}
Q_j\big(E_{s_1,s_2}(j')\big)\, =\, \begin{cases} 0 & \text { if $j' \in \mathcal{J}(s_1,s_2)$ satisfies $j'>j$,}\\
1 & \text { if $j' \in \mathcal{J}(s_1,s_2)$ satisfies $j'\leq j$.}
\end{cases}\label{monoQ}
\end{equation}
Then the most general polynomial with the FOEL property is a nonnegative combination of the basis
polynomials $\{Q_j\}_{j \in \mathcal{J}(s_1,s_2)}$.
The actual polynomial may be found with the aid of Lagrange interpolation.
Define
$$
\mathcal{L}_{s_1,s_2}(z)\, =\, \prod_{j \in \mathcal{J}(s_1,s_2)} \big(z-E_{s_1,s_2}(j)\big)\, .
$$
Then,  we have
$$
Q_j(z)\, =\, \sum_{j' \leq j} \frac{\mathcal{L}_{s_1,s_2}(z)}{\mathcal{L}_{s_1,s_2}'\big(E_{s_1,s_2}(j')\big) \big(z-E_{s_1,s_2}(j')\big)}\, .
$$
For example, for $s_1=s_2=1$, we have $\mathcal{J}(s_1,s_2) = \{0,1,2\}$, with
$$
E_{1,1}(0)\, =\, 3\, ,\quad E_{1,1}(1)\, =\, 2\, ,\quad E_{1,1}(2)\, =\, 0\, .
$$
So
$$
Q_1(z)\, =\, \frac{1}{3} z^2 - \frac{2}{3} z\, ,\quad Q_2(z)\, =\, -\frac{1}{6} z^2 + \frac{5}{6} z\, ,\quad Q_3(z)=1\, .
$$
Writing this in terms of the Heisenberg interaction, we have
\begin{align*}
Q_1(\mathfrak{h}_{1,1})\, &=\, \frac{1}{3} (\boldsymbol{S}_1 \cdot \boldsymbol{S}_2 )^2\, =\, P^{(0)}\, ,\\
Q_2(\mathfrak{h}_{1,1})\, &=\, -\frac{1}{6} (\boldsymbol{S}_1 \cdot \boldsymbol{S}_2)^2 - \frac{1}{2} \boldsymbol{S}_1 \cdot \boldsymbol{S}_2 + \frac{2}{3} \mathbbm{1}\,
=\, P^{(0)} + P^{(1)}\, ,
\end{align*}
and of course $Q_3(\mathfrak{h}_{1,1}) = \mathbbm{1} = P^{(0)} + P^{(1)} + P^{(2)}$, where we denote $P^{(j)}$ to be the projection onto the spin $j$ subspace.
Note that $Q_2(\mathfrak{h}_{1,1})$ is the negative of the famous AKLT model \cite{affleck:1988}, because the ground state of that antiferromagnetic model is precisely the range
of $P^{(0)} + P^{(1)}$, for two sites.

One might naively hope that since one can completely resolve the question for $L=2$ of which interactions satisfy FOEL,
that this would lead to a complete resolution for larger $L$, as well.
But we remind the reader that quantum spin systems are difficult precisely because the nearest-neighbor interaction terms typically do not commute.
We do not expect to prove 
that all Hamiltonians with two-spin interactions of the form (\ref{monoQ}) satisfy FOEL. 
What we can prove is that all 
nearest neighbor Hamiltonians in one dimension in which each nearest neighbor term is a
sum with non-positive coefficients of a particular basis of polynomials,
which we call the {\em cascade basis} (see Lemma \ref{lem:cascadebasis})
and which are of course non-decreasing on the set
$\{\lambda_j\mid j=0,\ldots, 2s\}$, satisfy FOEL for chains of arbitrary length.

For the sake of comparison, we remark that in the case of two identical spins, the spectrum of the cascade basis operators, defined in (\ref{cascadebasisoperators}) below, is 
\begin{equation*}
	\lambda_j = \begin{cases} \frac{(2j +k)!(2s-k)!}{(2j-k)!(2s+k)!} & \textrm{ when } j \geq k/2 \\ 0 & \textrm{ otherwise } \end{cases}
\end{equation*}
where again, $j$ refers to the total spin $j$ irrep and $0 \leq k \leq 2s$. 
Thus, we see that there is degeneracy at $0$, but non-zero eigenvalues are simple.

\begin{figure}\label{fig:spin1}
\begin{tikzpicture}
%
%
\draw[very thick,->] (0,0) -- (0,4);
\draw[very thick,->] (0,0) -- (4,0);
\draw[above right] (0,3.5) node(1){$J_2$};
\draw[above right] (3.5,0) node(1){$J_1$};
\draw[very thick,-] (0,0) circle (3);
\filldraw[fill=gray!80,draw=black!100] 
(0,0) -- (-3,0) arc (180:201:3) -- cycle;
\draw[very thick,-] (-3,0) arc (180:201:3);
\filldraw[fill=gray!50,draw=black!100] 
(0,0) -- (0,2.4) arc (90:201:2.4) -- cycle;
\draw[very thick,-] (0,2.4) arc (90:201:2.4);
\filldraw[fill=gray!20,draw=black!100] 
(0,0) -- (0,1.8) arc (90:225:1.8) -- cycle;
\draw[very thick,-] (0,1.8) arc (90:225:1.8);
\end{tikzpicture}
\caption{The most general $SU(2)$-invariant nearest neighbor interaction for a spin-1 chain
is of the form  $J_1{\bf S}_x\cdot {\bf S}_{x+1}+J_2 ({\bf S}_x\cdot {\bf S}_{x+1})^2$. The ground
state of a translation invariant model with such interactions depends only on the signs of the
coupling constants $J_1$ and $J_2$ and their relative magnitude. In this pie diagram, the region
where the ground state is a saturated ferromagnet (maximal total spin) is indicated in light gray.
The region in middle gray is where the system with just two spins exhibits FOEL. In this paper,
as a direct application of Theorem \ref{thm:MainVersionThm1}, we prove that spin-1 chains of
arbitrary length with Hamiltonians of the form $H_L=\sum_{x=1}^L
J_1(x){\bf S}_x\cdot {\bf S}_{x+1}+J_2(x) ({\bf S}_x\cdot {\bf S}_{x+1})^2$ satisfy FOEL if for
$1\leq x\leq L-1$, $(J_1(x),J_2(x))$ belongs to the dark gray segment.}
\end{figure}
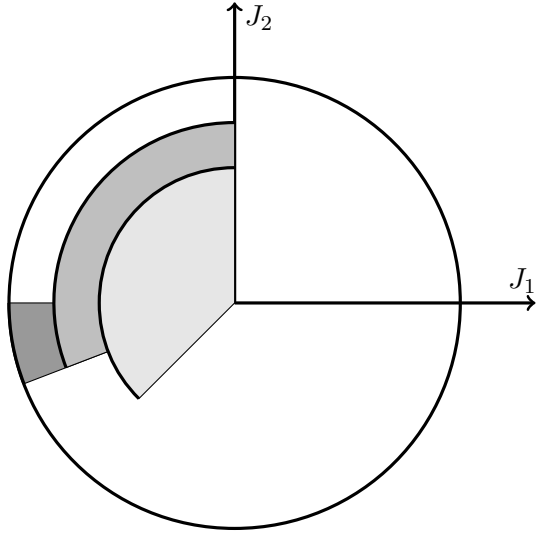

\subsection{Cascade basis}

Given $n \in \N$,
consider the tensor product $(\C^2)^{\otimes n}$.
The set of symmetric tensors forms an irreducible representation of $\text {SU}(2)$
of dimension $n+1$.
Let $T_n : (\C^2)^{\otimes n} \to \C^{n+1}$ be the symmetrization map given by 
\begin{equation*}
	T_n (v^{s_1} \otimes \cdots \otimes v^{s_n}) = v^{|s|}
\end{equation*}
where $v^{s_i}$ denotes the basis element which diagonalizes the third component of spin and has eigenvalue (or weight) $s_i$, and $|s| = s_1 + s_2 + \cdots + s_n$.
One can show also that $T_n^* : \C^{n+1} \to (\C^2)^{\otimes n }$ is given by 
\begin{equation*}
	T_n^* (v^m) = \binom{n}{\frac{n-m}{2}}^{-1} \sum_{\substack{s_1,\ldots,s_n \\ |s|=m}} v^{s_1} \otimes \cdots \otimes v^{s_n}
\end{equation*}

Given two integers $m,n$, we may consider 
$$
T_m^* \otimes T_n^* : \C^{m+1} \otimes \C^{n+1} \to (\C^2)^{\otimes (m+n)}\, .
$$
Here we identify the tensor product $(\C^2)^{\otimes m} \otimes (\C^2)^{\otimes n}$ with $(\C^2)^{\otimes (m+n)}$,
keeping the left-to-right order of the tensor factors:
$$
(v_1 \otimes \cdots \otimes v_n) \otimes (w_1 \otimes \cdots \otimes w_m) \mapsto v_1 \otimes \cdots \otimes v_n \otimes w_1 \otimes \cdots \otimes w_m\, .
$$
Given $k \in \{-m,-m+1,\dots,n\}$, we define
$$
U_{m,n}^{m+k,n-k}\, =\, (T_{m+k} \otimes T_{n-k}) (T_m^* \otimes T_n^*)\, ,
$$
which is an intertwiner
from $\C^{m+1} \otimes \C^{n+1}$ to $\C^{m+k+1} \otimes \C^{n-k+1}$.
Finally, we define
\begin{equation}
K_{m,n}(k) = (U_{m,n}^{m+k,n-k})^* U_{m,n}^{m+k,n-k} : \C^{m+1} \otimes \C^{n+1} \to \C^{m+1} \otimes \C^{n+1}\, .
\label{cascadebasisoperators}
\end{equation}
The following lemma will be proved in the next section.
\begin{lemma}
\label{lem:cascadebasis}
If $m\geq n$ then $\{K_{m,n}(k)\}_{k=0}^n$ forms a basis for the set of $\textrm{SU}(2)$-intertwiners
from $\C^{m+1} \otimes \C^{n+1}$ to itself.
Also, if we define 
$$
\tau : \C^{m+1} \otimes \C^{n+1} \to \C^{n+1} \otimes \C^{m+1}\, ,
$$
such that $\tau (v \otimes w) = w \otimes v$, then
$$
\tau^* K_{n,m}(-k) \tau\, =\, K_{m,n}(k)\, ,
$$
for all $m,n \in \N$ and $k = -m,\dots,n$.
\end{lemma}

In particular, since the $K_{m,n}(k)$ form a basis for the $\textrm{SU}(2)$-intertwiners, we observe that $H$ can be written in terms of the $K_{m,n}(k)$,
\begin{equation}
	H = \sum_{i=1}^{L-1} \sum_{k=0}^{\min(n_i,n_{i+1})} J_k^{(i)} K_{n_i,n_{i+1}}(k)
	\label{eqn:cascadeHam}
\end{equation}
Our interest in the cascade basis is that in the cascade basis, there is a particularly nice set of inequalities in the coupling coefficients where it is possible to prove FOEL.
Our main result is as follows.

\begin{thm}
\label{thm:MainVersionThm1}
Let $H$ be given as in equation (\ref{eqn:cascadeHam}). If $J_k^{(i)} \leq 0 $ for all $1\leq k$, then  
the full Hamiltonian $H$ satisfies the FOEL property.
\end{thm}

We use Frenkel and Khovanov's graphical representation to prove this theorem and note that it also applies to quantum spin chains with a symmetry of the quantum group $\text {SU}_q(2)$ for $q>0$. In what follows, we will provide a set of sufficient conditions for FOEL. As it turns out, the cascade basis allows us to demonstrate that these sufficient conditions are satisfied by an easy set of inequalities on the coupling coefficients $J_k^{(i)}$. 

Many bases of $\textrm{SU}(2)$-intertwiners of $\C^{m+1} \otimes \C^{n+1}$ have been considered up to now:
\begin{itemize}
	\item Powers of the Heisenberg interaction, $(S_i \cdot S_{i+1})^k$
	\item Projections on to total spin $j$ components
	\item The Indicator basis, $X_j$
	\item The Temperley Lieb Basis, which we describe after introducing the graphical calculus
	\item The Cascade basis, $K_{m,n}(k)$
\end{itemize}
It is possible to calculate explicit change of basis formulas, and this allows one to obtain 
inequalities for FOEL in other bases. Expressing the various bases as linear combinations
of the projections onto the subspaces of fixed total spin is equivalent to finding the spectral
decomposition for each of the bases, which is also perhaps more directly useful.
For example, for a pair of spin 1's we find the following expressions.
\begin{align*}
&X_0 = P^{(0)} + P^{(1)} + P^{(2)} =\unity ,\quad X_1 =  P^{(1)} + P^{(2)},\quad X_2 =   P^{(2)} \\
& K_{2,2}(0) =  P^{(0)} + P^{(1)} + P^{(2)} ,\quad  K_{2,2}(1) =  \frac13 P^{(1)} + P^{(2)}, 
\quad K_{2,2}(2) = P^{(2)}.
\end{align*}
For the Temperley-Lieb basis we use a graphical notation introduced in \cite{kauffman:1994}.
The spectral decompositions are
$$
\begin{minipage}{3cm}
\begin{tikzpicture}[xscale=0.75,yscale=0.75,bend angle=90,
spin/.style={rectangle,draw,thick,minimum width=10mm},scale=1]
\node at (-1,1) [spin] (a1) {};
\node at (1,1) [spin] (b1) {};
\node at (-1,-1) [spin] (a2) {};
\node at (1,-1) [spin] (b2) {};
\draw (a1.-150) .. controls +(0,-0.5) and +(0,0.5) .. node[left] {\small $$} (a2.150);
\draw (a1.-30) .. controls +(0,-0.5) and +(0,0.5) .. node[left] {\small $$} (a2.30);
\draw (b1.-30) .. controls +(0,-0.5) and +(0,0.5) .. node[right] {\small $$} (b2.30);
\draw (b1.-150) .. controls +(0,-0.5) and +(0,0.5) .. node[left] {\small $$} (b2.150);
\end{tikzpicture}
\end{minipage}
= P^{(0)} + P^{(1)} + P^{(2)}, \quad 
\begin{minipage}{3cm}
\begin{tikzpicture}[xscale=0.75,yscale=0.75,bend angle=90,
spin/.style={rectangle,draw,thick,minimum width=10mm},scale=1]
\node at (-1,1) [spin] (a1) {};
\node at (1,1) [spin] (b1) {};
\node at (-1,-1) [spin] (a2) {};
\node at (1,-1) [spin] (b2) {};
\draw (a1.-150) .. controls +(0,-0.5) and +(0,0.5) .. node[left] {\small $$} (a2.150);
\draw (b1.-30) .. controls +(0,-0.5) and +(0,0.5) .. node[right] {\small $$} (b2.30);
\draw (a1.-30) .. controls +(0,-0.75) and +(0,-0.75) .. node[above] {\small $$} (b1.-150);
\draw (a2.30) .. controls +(0,0.75) and +(0,0.75) .. node[below] {\small $$} (b2.150);
\end{tikzpicture}
\end{minipage}
= -P^{(1)} - \frac32 P^{(2)},\quad 
$$
$$
\begin{minipage}{3cm}
\begin{tikzpicture}[xscale=0.75,yscale=0.75,bend angle=90,
spin/.style={rectangle,draw,thick,minimum width=10mm},scale=1]
\node at (-1,1) [spin] (a1) {};
\node at (1,1) [spin] (b1) {};
\node at (-1,-1) [spin] (a2) {};
\node at (1,-1) [spin] (b2) {};
\draw (a1.-30) .. controls +(0,-0.50) and +(0,-0.50) .. node[above] {\small $$} (b1.-150);
\draw (a2.30) .. controls +(0,0.50) and +(0,0.50) .. node[below] {\small $$} (b2.150);
\draw (a1.-150) .. controls +(0,-0.90) and +(0,-0.90) .. node[above] {\small $$} (b1.-30);
\draw (a2.150) .. controls +(0,0.90) and +(0,0.90) .. node[below] {\small $$} (b2.30);
\end{tikzpicture}
\end{minipage}
= 3 P^{(2)}.
$$
The projections $P^{(0)},P^{(1)}$, and $P^{(2)}$ can in turn be expressed as a polynomial
in the Heisenberg interaction:
\begin{eqnarray*}
P^{(0)}&=& -\frac{1}{3}\unity + \frac{1}{3} ({\bf S}_x\cdot {\bf S}_{x+1})^2\\
P^{(1)}&=& \unity - \frac{1}{2}{\bf S}_x\cdot {\bf S}_{x+1} -\frac{1}{2} ({\bf S}_x\cdot {\bf S}_{x+1})^2\\
P^{(2)}&=& \frac{1}{3}\unity + \frac{1}{2}{\bf S}_x\cdot {\bf S}_{x+1} 
+\frac{1}{6} ({\bf S}_x\cdot {\bf S}_{x+1})^2
\end{eqnarray*}
Theorem \ref{thm:MainVersionThm1} states that we have FOEL for spin-1 chains of arbitrary 
length if all interactions are of the form 
\begin{equation*}
	J_1(i) S_i \cdot S_{i+1} + J_2 (i)S_i \cdot S_{i+1})^2 =  c_0(i) K_{1,1}(0) +
	c_1(i) K_{1,1}(1) + c_2(i) K_{1,1}(2) 
\end{equation*}
with $c_0(i),c_1(i),c_2(i)\geq 0$, for all $i$. Using the spectral decompositions above, we
see that these conditions are equivalent to $J_2^{(i)} \leq 0$ and $\frac13 J_1^{(i)} \leq J_2^{(i)}$
for all $1 \leq i \leq L-1$.

\section{Probabilistic interpretation}
\label{sec:prob}

In this section, we want to briefly summarize the implications of FOEL for models from 
probability theory. It is well known that the Markov generator for the symmetric exclusion 
process on a finite graph is equivalent via a similarity transformation to the negative of the 
quantum Hamiltonian for the XXX Heisenberg ferromagnet on the same graph.
For reviews on this topic, see \cite{Thomas,Caputo,nachtergaele:2005} and for other related
models see \cite{AlcarazDrozHenkelRittenberg}.

To make the connection with the more general quantum spin Hamiltonians we study
in this paper, consider an urn mixing model of the following type. 
Suppose that one has $L$ urns, indexed in left-to-right order by $x=1,\dots,L$, each holding
$n$ balls, which come in two colors, say red and white.
At time $t\geq 0$, the urn at $x$ contains $k_x(t)$ red balls and $n-k_x(t)$ white balls.
So the state space for this Markov chain is $\{0,1,\dots,n\}^L$,
specifying all the $k_x(t)$'s for $x=1,\dots,L$.

The dynamics is specified by giving Poisson rates $\lambda_{x,x+1} \geq 0$ for each
pair of urns $\{x,x+1\} \subseteq \{1,\dots,L\}$,
and by specifying a probability distribution $\rho_{x,x+1}$ on the set $\{0,1,\dots,n\}$.

We suppose that there are independent Poisson processes $\{\tau_{x,x+1}(1),\tau_{x,x+1}(2),\dots\}$
for each pair $\{x,x+1\}$, with rate $\lambda_{x,x+1}$.
We also suppose that, independent of this, there are independent random variables $N_{x,x+1}(j)$
for each pair $\{x,x+1\}$ and each $j \in \N$, such that $N_{x,x+1}(j)$ is distributed according to $\rho_{x,x+1}$.

At time $\tau_{x,x+1}(j)$, one removes $N_{x,x+1}(j)$ balls from urn $x$,
uniformly at random among all $n$-choose-$N_{x,x+1}(j)$ possibilities,
and also removes $N_{x,x+1}(n)$ balls from urn $x+1$ uniformly among the $n$-choose-$N_{x,x+1}(j)$
possibilities.
Then these two groups of $N_{x,x+1}(j)$ balls are exchanged and replaced in the opposite urns.

The random dynamics of the model is symmetric for each urn, in the sense that all $n$ balls
in a single urn may be permuted in any fixed (non-random) way, without changing the law of 
the dynamics. Moreover, one may construct a graphical representation for this model.
The existence of the graphical representation means that one can prescribe {\it a priori} which balls
are exchanged, using the graphical representation, and afterwards may specify their initial colors.
This accounts for the sometimes mis-understood hidden $\textrm{SU}(2)$ symmetry of the 
symmetric exclusion process. This urn mixing process, being a generalization of the symmetric 
exclusion process also possesses the $\textrm{SU}(2)$ symmetry.
(If we allowed balls of $k$ colors instead of just two, then the symmetry group would be 
$\textrm{SU}(k)$.)

One may make a conjecture analogous to Aldous's conjecture for the symmetric exclusion process.
Note that the dynamics specified preserves the total number of red and white balls, separately.
It merely changes their positions.

\begin{conj}
For $1\leq k\leq nL-1$, let $\gamma_k$ be the spectral gap of this model,
when restricted to the state space consisting of all possible configurations with exactly $k$ red balls
and $nL-k$ white balls. (For $k=0$ or $nL$ there is only one such
configuration and therefore no spectral gap.)
Then $\gamma_k$ is independent of $k$: in other words it is equal to $\gamma_1$
which is the spectral gap for the associated random walk.
\end{conj}

As a special case of our FOEL theorem, we can partially verify this conjecture as follows.

\begin{cor}
For each $k=0,\dots,n$,
let $\rho_k$ be the distribution on $\{0,\dots,n\}$
which is the hypergeometric distribution:
$$
\rho_k(\{j\})\, =\, \frac{\binom{k}{j}\binom{n}{k-j}}{\binom{n+k}{k}}\, .
$$
Then the conjecture above is satisfied as long as $\rho_{x,x+1}$ is a convex mixture of $\rho_0,\dots,\rho_n$
for each pair $\{x,x+1\}$.
\end{cor}

It is not surprising that such models would be related to quantum spin systems.
A common way to understand the quantum nature of a quantum spin chain is to use the 
Trotter product formula to map the quantum system to a classical spin system evolving in 
time, sometimes called the ``imaginary time'' construction (see. e.g.,
 \cite{aizenman:1994,nachtergaele:1994}
and the recent review \cite{GoldschmidtUeltschiWindridge}.)

\section{Graphical Calculus for $\U$}\label{sec:graphical_calculus}

In our calculations in large tensor product spaces, the standard tensor notation becomes unwieldy, and it is useful to switch to graphical methods. We review the methods and results detailed in \cite{frenkel:1997}. 
We define $\U$ to be the associative algebra over $\comps (q)$ generated by $K,K^{-1}, E,F$ subject to the following relations:
\begin{align*}
	K^{\pm1} K^{\mp1} & = \id\\
	KE &= q^2 EK \\
	KF & = q^{-2} FK\\
	[E,F] &= \frac{K-K^{-1}}{q-q^{-1}}
\end{align*}
The symbol $q$, unless otherwise indicated, will be an indeterminate. However, for our results using the Perron-Frobenius theorem, we will need to specialize to $q \in \reals$ such that $0<q$. 

$\U$ will act on tensor product representations via the coproduct:
\begin{align*}
	\Delta(K^{\pm1}) & = K^{\pm1} \otimes K^{\pm1}\\
	\Delta(E) &= E \otimes 1 + K \otimes E \\
	\Delta(F) &= F \otimes K^{-1} + 1 \otimes F
\end{align*}

The irreducible representations of $\U$, $V(n)$ for $0\geq n$, are given as the vector space generated by $\{v^{-n}, v^{-n+2}, \ldots, v^{n-2}, v^n\}$ with the following $\U$ action:
\begin{align*}
	K^{\pm1} v^m & = q^{\pm m} v^m \\
	E v^m & = \left[\frac{n-m}{2}\right] v^{m+2}\\
	F v^m & = \left[\frac{n+m}{2}\right] v^{m-2}
\end{align*}
where $[n] = q^{-n+1} + q^{-n+3} + \cdots + q^{n-3} + q^{n-1} = \frac{q^n-q^{-n}}{q-q^{-1}}$. Graphically, one represents this representation as 
\begin{equation*}
	v^m = 
	\raisebox{-.25cm}{
\begin{tikzpicture}
	\node (n1) at (0,0) [draw,thick,rectangle,minimum height=.5cm,minimum width=1.75cm] {$n$};
	\draw[<-] (n1.160) -- node[anchor=west] {\tiny{$\cdots$}} ++(0,.5); 
	\draw[<-] (n1.105) -- ++(0,.5); 
	\draw[->] (n1.85) --  ++(0,.5); 
	\draw[->] (n1.20) -- node[anchor=east] {\tiny{$\cdots$}} ++(0,.5); 
\end{tikzpicture}
} \textrm{ where } m = \# \,
\raisebox{-.25cm}{
\begin{tikzpicture}
	\draw[thick,->] (0,0) -- (0,.5);
\end{tikzpicture} 
}
- \# \,
\raisebox{-.25cm}{
\begin{tikzpicture}
	\draw[thick,<-] (0,0) -- (0,.5);
\end{tikzpicture} 
}
\end{equation*}
By convention, for $V(1)$, we do not draw a box. $V(n)$ can be realized as the $q$-symmetric subset in $V(1)^{\otimes n}$, and thus the box in our picture represents not only the irreducible representation $V(n)$, but also the operation of $q$-symmetrization. This is made explicit through the notion of the Jones-Wenzl projector, which is heuristically a projection from $V(1)^{\otimes n}$ to $V(n)$ composed with an injection from $V(n)$ to $V(1)^{\otimes n}$ such that their composition is precisely the operation of $q$-symmetrization.
Let $s_i\in \{-1,+1\}$, $s=(s_1,s_2, \ldots, s_n)$, $|s|= \sum_i s_i$, $\|s\|_+ = \sum_{i<j} \{s_i > s_j\}$, and $\|s\|_- = \sum_{i<j} \{s_i<s_j\}$ where $\{a>b\} = 1$ if $a>b$ and 0 if $a\leq b$. Then define the maps $T^*_n : V(n) \to V(1)^{\otimes n}$ and $T_n: V(1)^{\otimes n} \to V(n)$ to be given by
\begin{align*}
	&T^*_n (v^m) = {\qbinom{n}{\frac{n-m}{2}}}^{-1} \sum_{s,|s|=m} q^{\|s\|_-} v^{s_1} \otimes v^{s_2} \otimes \cdots \otimes v^{s_n} \\
	&T_n (v^{s_1} \otimes \cdots \otimes v^{s_n}) = q^{\|s\|_+} v^{|s|}
\end{align*}
The Jones-Wenzl projector, denoted by $p_n$, is then 
\begin{equation*}
	p_n = T^*_n \circ T_n = 
	\raisebox{-.75cm}{
	\begin{tikzpicture}
		\node (n1) at (0,0) [draw,thick,rectangle,minimum height=.5cm,minimum width=1cm] {$n$};
		\draw[] (n1.140) -- node[anchor=west] {\tiny{$\cdots$}} ++(0,.5);
		\draw[] (n1.40) -- ++(0,.5);
		\draw[] (n1.-140) -- node[anchor=west] {\tiny{$\cdots$}} ++(0,-.5);
		\draw[] (n1.-40) -- ++(0,-.5);
	\end{tikzpicture}
	} = 
	\raisebox{-.65cm}{
	\begin{tikzpicture}
		\node(n1) at (0,0) [draw,thick,rectangle,minimum width=1cm] {};
		\draw (n1) -- node[auto] {$n$} ++(0,.75);
		\draw (n1) -- ++(0,-.75);
	\end{tikzpicture}
	}
\end{equation*}
where in the second diagram, we have chosen to omit the label for the box when the context is clear. 
\subsection{The Algebra of Intertwiners, $\textbf{End}_{\U}(V(1)^{\otimes n})$}
It is a well known fact first described by Jimbo \cite{jimbo:1986} that the Temperley-Lieb Algebra, denoted by $TL_n$, is the algebra of intertwiners for the $\U$ action. It is generated in the following way: first we define the two intertwiners $\delta:V(0) \to V(1)\otimes V(1)$ and $\varepsilon:V(1) \otimes V(1) \to V(0)$ such that
\begin{equation*}
	\delta(1) = \down \up - q^{-1} \up \down = 
	\begin{tikzpicture}
		\draw[] (0,0) arc (180:0:.25);
	\end{tikzpicture}
\end{equation*}
\begin{align*}
	\varepsilon(v^1\otimes v^1) =& \varepsilon(v^{-1} \otimes v^{-1}) = 0 = 
	\raisebox{-.25cm}{
	\begin{tikzpicture}
		\draw[] (0,0) arc (-180:0:.25);
		\draw[->] (0,0) -- (0,.5);
		\draw[->] (.5,0) -- (0.5,.5);
	\end{tikzpicture}
	} = 
	\raisebox{-.25cm}{
	\begin{tikzpicture}
		\draw[] (0,0) arc (-180:0:.25);
		\draw[<-] (0,0) -- (0,.5);
		\draw[<-] (.5,0) -- (0.5,.5);
	\end{tikzpicture}
	} \\
	&\varepsilon(v^1\otimes v^{-1}) = -q = 
	\raisebox{-.25cm}{
	\begin{tikzpicture}
		\draw[] (0,0) arc (-180:0:.25);
		\draw[->] (0,0) -- (0,.5);
		\draw[<-] (.5,0) -- (0.5,.5);
	\end{tikzpicture}
	} \\
	&\varepsilon(v^{-1} \otimes v^1) = 1 = 
	\raisebox{-.25cm}{
	\begin{tikzpicture}
		\draw[] (0,0) arc (-180:0:.25);
		\draw[<-] (0,0) -- (0,.5);
		\draw[->] (.5,0) -- (0.5,.5);
	\end{tikzpicture}
	} 
\end{align*}
\begin{equation*}
	 \varepsilon \circ \delta = -[2] = 
	 \raisebox{-.25cm}{
	 \begin{tikzpicture}
		 \draw[] (0,0) arc (360:0:.25);
	 \end{tikzpicture}
	 }
\end{equation*}
\begin{equation}
	(\varepsilon \otimes \id_{V(1)}) \circ (\id_{V(1)} \otimes \delta) = 
	\raisebox{-.5cm}{
	\begin{tikzpicture}
		\draw[] (0,.5) -- (0,0) arc (-180:0:.25) arc (180:0:.25) -- (1,-.5);
	\end{tikzpicture}
	} =
	\raisebox{-.5cm}{
	\begin{tikzpicture}
		\draw[] (0,.5) -- (0,-.5);
	\end{tikzpicture}
	} = \id_{V(1)} = 
	\raisebox{-.5cm}{
	\begin{tikzpicture}
		\draw[] (0,-.5) -- (0,0) arc (180:0:.25) arc (-180:0:.25) -- (1,.5);
	\end{tikzpicture}
	} = (\id_{V(1)} \otimes \varepsilon) \circ ( \delta \otimes \id_{V(1)})
		\label{eqn:isotopy}
	\end{equation}
One may interpret the arcs $\delta$ and $\varepsilon$ as $q$-antisymmetrizers. 

In contrast to \cite{frenkel:1997}, our diagrams are read from top to bottom, and that composition of diagrams is obtained by concatenation below. 
Equation (\ref{eqn:isotopy}) provides an isotopy relation which allows us to take a meandering line and pull it taut without changing the meaning of the diagram. Hence, if we define $U_i: V(1)^{\otimes n} \to V(1)^{\otimes n}$ by 
\begin{equation*}
	U_i = \id_1 \otimes \cdots \otimes \id_{i-1} \otimes (\varepsilon \circ \delta) \otimes \id_{i+2} \otimes \cdots \id_{n} = 
	\raisebox{-.75cm}{
	\begin{tikzpicture}
		\draw[] (0,1) -- node[anchor=west] {\tiny{$\cdots$}} (0,0) node[anchor=north] {\tiny{$1$}};
		\draw[] (.6,1) -- (.6,0);
		\draw[] (1,1) arc (-180:0:.25);
		\draw[] (1,0) node[anchor=north] {\tiny{$i$}} arc (180:0:.25) node[anchor=north] {\tiny{$i+1$}};
		\draw[] (1.9,1) -- node[anchor=west] {\tiny{$\cdots$}} (1.9,0);
		\draw[] (2.5,1) -- (2.5,0) node[anchor=north] {\tiny{$n$}};
	\end{tikzpicture}
	}
\end{equation*}
Now with the isotopy relations and the evaluation of the bubble, we see that $U_1, U_2, \ldots, U_{n-1}$ are the generators of the Temperley-Lieb algebra subject to the relations:
\begin{align*}
	U_i U_{i\pm1} U_i & = U_i \\
	U_i^2 & = -[2] U_i \\ 
	[U_i,U_j] &= 0 \textrm{ for } |i-j|>1
\end{align*}
There is a convenient basis of $TL_n$ called the dual canonical basis of $TL_n$, denoted $B^{TL}_n$, which consists of all diagrams of $2n$ points, $n$ points in a row above and $n$ points in a row below, such that points are pairwise connected by a line in such a fashion that lines do not intersect. The Jones-Wenzl projector intertwines the $\U$ action, since in the Clebsch-Gordan series for $V(1)^{\otimes n} p_n$ is the projection to the unique copy of the irreducible representation $V(n)$.  Thus, it can be represented as a sum of elements from $B^{TL}_n$. 
We now review useful properties of the Jones-Wenzl projector.
\subsection{The Jones-Wenzl Projector and Formulas}
The Jones-Wenzl projectors, $p_n$, satisfy some very nice conditions that make them particularly useful in our calculations. 
As a property of the interpretation as $q$-symmetrization and $q$-antisymmetrization, we have
\begin{equation}
	\raisebox{-.75cm}{
	\begin{tikzpicture}
		\node (n1) at (0,0) [draw,thick,rectangle,minimum width=1cm] {};
		\draw[] (n1.90) -- node[auto] {$n$} ++(0,.5);
		\draw[] (n1.-140) -- node[auto,swap] {$n-2$} ++(0,-.5);
		\draw[] (n1.-90) to [bend right=90] (n1.-20);
	\end{tikzpicture}
	} = 
	\raisebox{-.75cm}{
	\begin{tikzpicture}
		\node (n1) at (0,0) [draw,thick,rectangle,minimum width=1cm] {};
		\draw[] (n1.-90) -- node[auto] {$n$} ++(0,-.5);
		\draw[] (n1.40) -- node[auto,swap] {$n-2$} ++(0,.5);
		\draw[] (n1.160) to [bend left=90] (n1.90);
	\end{tikzpicture}
	} = 0
	\label{eqn:antisymm}
\end{equation}
For $m\leq n$, we have
\begin{equation}
	\raisebox{-1cm}{
	\begin{tikzpicture}
		\node (n1) at (0,0) [draw,thick,rectangle,minimum width=1cm] {};
		\draw[] (n1.-90) -- node[auto] {$n$} ++(0,-.5);
		\draw[] (n1.30) -- node[auto,swap] {$n-m$} ++(0,1);
		\node (n2) at (-.20,.75) [draw,thick,rectangle,minimum width=.6cm] {};
		\draw[] (n1.135) -- ++(0,.5);
		\draw[] (n1.135) ++(0,.75) -- ++(0,.25);
	\end{tikzpicture}
	} =  
	\raisebox{-1cm}{
	\begin{tikzpicture}
		\node (n1) at (0,0) [draw,thick,rectangle,minimum width=1cm] {};
		\draw[] (n1.90) -- node[auto] {$n$} ++(0,.5);
		\draw[] (n1.-55) -- ++(0,-.5);
		\node (n2) at (.20,-.75) [draw,thick,rectangle,minimum width=.6cm] {};
		\draw[] (n1.-150) -- node[auto,swap] {$n-m$} ++(0,-1);
		\draw[] (n1.-55) ++(0,-.75) -- ++(0,-.25);
	\end{tikzpicture}
	} = 
	\raisebox{-.75cm}{
	\begin{tikzpicture}
		\node(n1) at (0,0) [draw,thick,rectangle,minimum width=1cm] {};
		\draw[] (n1.-90) -- node[auto] {$n$} ++(0,-.5);
		\draw[] (n1.90) -- ++(0,.5);
	\end{tikzpicture}
	}
	\label{eqn:symm}
\end{equation}

The projectors, $p_n$, have positive coefficients when expanded into the Temperley-Lieb basis.

\begin{thm}
	We have the decomposition
	\begin{equation}
		[n]! p_n = \sum_{d\in B^{TL}_n} P(d) d
		\label{eqn:JWPosIntDecomp}
	\end{equation}
	where the coefficients $P(d) \in q^{n(n-1)/2} \nats [q^{-1}]$.
\end{thm}
\begin{proof}
See \cite{frenkel:1997}.
\end{proof}


\begin{lemma}[Wenzl Relation]
\label{lem:JonesWenzl}
For each $n \in \N$,
\begin{align*}
\begin{minipage}{1.25cm}
\begin{tikzpicture}[xscale=1,yscale=1,bend angle=90,
spin/.style={rectangle,draw,thick,minimum width=10mm},scale=1]
\node at (0,0) [spin] (n) {};
\draw (n.90) -- node[left] {$n$} +(0,0.75);
\draw (n.-90) -- node[left] {$n$} +(0,-0.75);
\end{tikzpicture}
\end{minipage}\quad  
&= \begin{minipage}{2.5cm}
	\begin{tikzpicture}[xscale=1,yscale=1,bend angle=90,
		spin/.style={rectangle,draw,thick,minimum width=10mm},scale=1]
		\node at (0,0) [spin] (n) {};
		\draw (n.90) -- node[left] {$n-1$} +(0,0.75);
		\draw (n.-90) -- node[left] {$n-1$} +(0,-0.75);
		\draw (1,1) -- node[right] {$1$} (1,-1);
	\end{tikzpicture}
\end{minipage} \quad
+ \frac{[n-1]}{[n]} \quad
\begin{minipage}{3.5cm}
\begin{tikzpicture}[xscale=1,yscale=1,
	spin/.style={rectangle,draw,thick,minimum width=10mm},scale=1]
\node at (0,0.75) [spin] (n1) {};
\node at (0,-0.75) [spin] (n2) {};
\draw (n1.90) -- node[above=2mm, left=2mm] {$n-1$} +(0,0.5);
\draw (n2.-90) -- node[below=2mm, left=2mm] {$n-1$} +(0,-0.5);
\draw (n1.-60) .. controls +(0.5,-0.5) and +(0,-0.5) ..  node[above=2mm, right=3mm] {$1$} +(0.8,0.6);
\draw (n2.60) .. controls +(0.5,0.5) and +(0,0.5) ..  node[below=2mm, right=3mm] {$1$} +(0.8,-0.6);
\draw (n1.-150) -- node[right] {$n-2$} (n2.150); 
\end{tikzpicture}
\end{minipage} \\
&= \quad
\begin{minipage}{2.5cm}
\begin{tikzpicture}[xscale=1,yscale=1,bend angle=90,
spin/.style={rectangle,draw,thick,minimum width=10mm},scale=1]
		\node at (0,0) [spin] (n) {};
		\draw (n.90) -- node[left] {$n-1$} +(0,0.75);
		\draw (n.-90) -- node[left] {$n-1$} +(0,-0.75);
		\draw (1,1) -- node[right] {$1$} (1,-1);
\end{tikzpicture}
\end{minipage} \quad
+  \frac{[n-1]}{[n]} \quad 
\begin{minipage}{3.5cm}
\begin{tikzpicture}[xscale=1,yscale=1,bend angle=90,
spin/.style={rectangle,draw,thick,minimum width=8mm},scale=1]
\node at (0,0.75) [rotate=90,spin] (n1) {};
\node at (0,-0.75) [rotate=-90,spin] (n2) {};
\draw (n1.90) .. controls +(-0.5,0) and +(0,-0.4) .. node[above=2mm, left=2mm] {$n-1$} +(-0.75,0.7);
\draw (n2.-90) .. controls +(-0.5,0) and +(0,0.4) .. node[below=2mm, left=2mm] {$n-1$} +(-0.75,-0.7);
\draw (n1.-150) .. controls +(0.5,0) and +(0.5,0) .. node[right] {$n-2$} (n2.150); 
\draw (n1.-60) .. controls +(0.25,0) and +(0,-0.5) ..  node[above=2mm, right=3mm] {$1$} +(0.8,0.6);
\draw (n2.60) .. controls +(0.25,0) and +(0,0.5) ..  node[below=2mm, right=3mm] {$1$} +(0.8,-0.6);
\end{tikzpicture}
\end{minipage} \quad .
\end{align*}
where the last equality follows by the isotopy relation given in equation (\ref{eqn:isotopy}). 
\end{lemma}

\begin{proof}
See, for example, \cite{kauffman:1994}.
\end{proof}

The following lemma appears in \cite{frenkel:1997} as well as in \cite{kim:2007} where it is called
a ``single clasp expansion.''

\begin{lemma}
\label{lem:singleclasp}
For each $n \in \N$,
\begin{equation*}
\begin{minipage}{2.25cm}
\begin{tikzpicture}[xscale=1,yscale=1,bend angle=90,
spin/.style={rectangle,draw,thick,minimum width=9mm},scale=1]
\node at (0,0) [rotate=-30,spin] (n) {};
\draw (n.135) .. controls +(0.1,0.25) and +(0,-0.5) .. node[above=2mm,left=0mm] {$n-1$} +(0.3,1);
\draw (n.45) .. controls +(0.5,0.4) and +(0.5,0.75) .. node[below=2mm,right=1mm] {$1$} +(0.25,-1);
\draw (n.-90) .. controls +(-0.1,-0.2) and +(0,0.3) .. node[below=1mm,left=0mm] {$n$} +(-0.1,-0.85);
\end{tikzpicture}
\end{minipage} \quad
= \quad
\sum_{k=1}^{n} \frac{[k]}{[n]} \quad 
\begin{minipage}{4cm}
\begin{tikzpicture}[xscale=1,yscale=1,bend angle=90,
spin/.style={rectangle,draw,thick,minimum width=9mm},scale=1]
\node at (0,0) [spin] (n) {};
\draw (n.90) -- node[above=2mm,left=0mm] {$n-1$} +(0,1);
\draw (n.-30) -- node[right=1mm] {$n-k$} +(1,-1.5);
\draw (n.-150) -- node[left=1mm] {$k-1$} +(-1,-1.5);
\draw (-0.5,-1.5) .. controls +(0,0.5) and +(0,0.5) .. node[above] {$1$} (0.5,-1.5);
\end{tikzpicture}
\end{minipage} \quad .
\end{equation*}
\end{lemma}

The following corollary was proved in \cite{masbaum:1994} and in \cite{frenkel:1997}.

\begin{cor}
\label{cor:MasbaumVogel}
For each $j,k \in \{0,1,\dots \}$,
\begin{equation*}
\begin{minipage}{3.5cm}
\begin{tikzpicture}[xscale=0.75,yscale=0.75,bend angle=90,
spin/.style={rectangle,draw,thick,minimum width=10mm},scale=1]
\node at (90:2.2cm) [rotate=0,spin] (a) {};
\node at (210:2.2cm) [rotate=-60,spin] (b) {};
\node at (330:2cm) [rotate=60,spin] (c) {};
\node at (0,0.5) [spin] (d) {};
\draw (a.-90) .. controls +(0,-0.5) and +(0,0.5) .. node[left] {$j+k$} (d.150);
\draw (d.-150) .. controls +(0,-0.5) and +(0.7,0.4) .. node[above left] {$k+1$} (b.90);
\draw (d.30) .. controls +(0.5,1.25) and +(-0.6,0.4) .. node[right] {$1$} (c.30);
\draw (d.-30) .. controls +(0,-0.5) and +(-0.7,0.3) .. node[left] {$j$} (c.150);
\end{tikzpicture}
\end{minipage} \quad 
= \
\frac{[k+1]}{[j+k+1]} \quad 
\begin{minipage}{3.5cm}
\begin{tikzpicture}[xscale=0.75,yscale=0.75,bend angle=90,
spin/.style={rectangle,draw,thick,minimum width=10mm},scale=1]
\node at (90:1.8cm) [rotate=0,spin] (a) {};
\node at (210:1.8cm) [rotate=-45,spin] (b) {};
\node at (330:1.8cm) [rotate=45,spin] (c) {};
\draw (a.-150) .. controls +(-0.1,-1) and  +(0.6,0.5) .. node[above left] {$k$} (b.150);
\draw (a.-30) .. controls +(0.1,-1) and +(-0.6,0.5) .. node[above right] {$j$} (c.30);
\draw (b.30) .. controls +(0.7,0.5) and +(-0.7,0.5) .. node[below] {$1$} (c.150);
\end{tikzpicture}
\end{minipage} \quad .
\end{equation*}
\end{cor}

A natural extension of this formula is the following.

\begin{prop}
\label{prop:FK}
For any $j,k, \ell \in \{0,1,2,\dots \}$, we have
\begin{equation*}
\begin{minipage}{3.5cm}
\begin{tikzpicture}[xscale=0.75,yscale=0.75,bend angle=90,
spin/.style={rectangle,draw,thick,minimum width=10mm},scale=1]
\node at (90:2.2cm) [rotate=0,spin] (a) {};
\node at (210:2.2cm) [rotate=-60,spin] (b) {};
\node at (330:2cm) [rotate=60,spin] (c) {};
\node at (0,0.5) [spin] (d) {};
\draw (a.-90) .. controls +(0,-0.5) and +(0,0.5) .. node[left] {$j+k$} (d.150);
\draw (d.-150) .. controls +(0,-0.5) and +(0.7,0.4) .. node[above left] {$k+\ell$} (b.90);
\draw (d.30) .. controls +(0.65,1.5) and +(-0.6,0.5) .. node[right] {$\ell$} (c.30);
\draw (d.-30) .. controls +(0,-0.5) and +(-0.7,0.3) .. node[left] {$j$} (c.150);
\end{tikzpicture}
\end{minipage} \quad 
= \
\frac{[j+k]!\, [k+\ell]!}{[k]!\, [j+k+\ell]!} \quad 
\begin{minipage}{3.5cm}
\begin{tikzpicture}[xscale=0.75,yscale=0.75,bend angle=90,
spin/.style={rectangle,draw,thick,minimum width=10mm},scale=1]
\node at (90:1.8cm) [rotate=0,spin] (a) {};
\node at (210:1.8cm) [rotate=-45,spin] (b) {};
\node at (330:1.8cm) [rotate=45,spin] (c) {};
\draw (a.-150) .. controls +(-0.1,-1) and  +(0.6,0.5) .. node[above left] {$k$} (b.150);
\draw (a.-30) .. controls +(0.1,-1) and +(-0.6,0.5) .. node[above right] {$j$} (c.30);
\draw (b.30) .. controls +(0.8,0.5) and +(-0.8,0.5) .. node[below] {$\ell$} (c.150);
\end{tikzpicture}
\end{minipage} \quad .
\end{equation*}
\end{prop}

This is Lemma 3.16 in \cite{frenkel:1997}.

\subsection{The Graphical Construction of the Cascade Basis}
Following the construction of the $K_{m,n}(k)$ for $SU(2)$, we generalize to the case of $SU_q(2)$ via the graphical calculus. Recall that the fundamental map in the construction is the projection from $V(1)^{\otimes n}$ to $V(n)$ that intertwines the $\U$ action, $T_n$. 
\begin{equation*}
	U_{m,n}^{m+k,n-k} = (T_{m+k} \otimes T_{n-k} ) (T_m^* \otimes T_n^*) = 
	\raisebox{-.75cm}{
	\begin{tikzpicture}
		[spin/.style={rectangle,draw,thick,minimum width=1cm},scale=1]
	\node(a) at (0,0) [rotate=0,spin] {};
	\node(b) at (1.5,0) [rotate=0,spin] {};
	\node(c) at (0,-1.5) [rotate=0,spin] {};
	\node(d) at (1.5,-1.5) [rotate=0,spin] {};
	\draw (a) -- node[auto,swap] {$m$} (c);
	\draw (c.35) -- node[auto] {$k$} (b.-145);
	\draw (b) -- node[auto] {$n-k$} (d);
\end{tikzpicture}
	}
\end{equation*}
and 
\begin{equation*}
	K_{m,n}(k) = (U_{m,n}^{m+k,n-k})^* U_{m,n}^{m+k,n-k} = 
	\raisebox{-1.5cm}{
\begin{tikzpicture}
		[spin/.style={rectangle,draw,thick,minimum width=1cm},scale=1]
	\node(a) at (0,0) [rotate=0,spin] {};
	\node(b) at (1.5,0) [rotate=0,spin] {};
	\node(c) at (0,-1.5) [rotate=0,spin] {};
	\node(d) at (1.5,-1.5) [rotate=0,spin] {};
	\node(e) at (0,-3) [rotate=0,spin] {};
	\node(f) at (1.5,-3) [rotate=0,spin] {};
	\draw (a) -- node[auto,swap] {$m$} (c);
	\draw (c.35) -- node[auto] {$k$} (b.-145);
	\draw (b) -- node[auto] {$n-k$} (d);
	\draw (c) -- node[auto,swap] {$m$} (e);
	\draw (c.-35) -- node[auto] {$k$} (f.145);
	\draw (f) -- node[auto,swap] {$n-k$} (d);
\end{tikzpicture}
	}
\end{equation*}
\subsection{Bases of $\textbf{End}_{\U}(V(m) \otimes V(n))$}

In this section, we justify Lemma \ref{lem:cascadebasis} in the context of $\U$ and remark that the results of the Lemma are obtained by taking the limit $q \to 1$.
Starting from the dual canonical basis of the Temperley-Lieb algebra, $TL_{m+n}=\textbf{End}_{\U}(V(1)^{\otimes m+n})$, we obtain a basis for $\textbf{End}_{\U}(V(m) \otimes V(n))$ by means of the surjective projection 
\begin{equation*}
P:\textbf{End}_{\U}(V(1)^{\otimes m+n})\to \textbf{End}_{\U},
\quad
P(\alpha)= (T_m \otimes T_n) \alpha (T_m^* \otimes T_n^*).
\end{equation*}
By equation (\ref{eqn:antisymm}) and simple combinatorial considerations, we obtain the basis 
\begin{equation*}
	\left\{ 
		\raisebox{-.75cm}{
	\begin{tikzpicture}
		\node(n1) at (0,0) [draw,thick,rectangle,minimum width=1cm] {};
		\node(n2) at (1.5,0) [draw,thick,rectangle,minimum width=1cm] {};
		\node(n1a) at (0,1.5) [draw,thick,rectangle,minimum width=1cm] {};
		\node(n2a) at (1.5,1.5) [draw,thick,rectangle,minimum width=1cm] {};
		\draw (n1.40) to [bend left=90] node[auto,swap] {{$k$}} (n2.140);
		\draw (n1a.-40) to [bend right=90] node[auto] {{$k$}} (n2a.-140);
		\draw (n1) -- node[auto] {$m-k$} (n1a);
		\draw (n2) -- node[auto,swap] {$n-k$} (n2a);
	\end{tikzpicture}
	}
	\right\}_{k=0}^{\min(m,n)} 
\end{equation*}
which we shall call the Temperley-Lieb basis. For the remainder of this section, let us make the assumption that $m\geq n$. We now can prove that the Cascade operators indeed form
a basis.

\begin{prop}
	The Cascade basis, $K_{m,n}(k)$ for $0\leq k \leq n$ is a basis of $\textbf{End}_{\U} (V(m) \otimes V(n))$. 
\end{prop}

\begin{proof}
	The proof is by finding an explicit change of basis from the Temperley-Lieb basis.
We proceed by expanding out the Jones-Wenzl projector of size $m+k$ in the operator $K_{m,n}(k)$ into a positive sum of elements in the Temperley-Lieb algebra. We obtain 
\begin{equation*}
	K_{m,n}(k) = 
	\raisebox{-1.5cm}{
\begin{tikzpicture}
		[spin/.style={rectangle,draw,thick,minimum width=1cm},scale=1]
	\node(a) at (0,0) [rotate=0,spin] {};
	\node(b) at (1.5,0) [rotate=0,spin] {};
	\node(c) at (0,-1.5) [rotate=0,spin] {};
	\node(d) at (1.5,-1.5) [rotate=0,spin] {};
	\node(e) at (0,-3) [rotate=0,spin] {};
	\node(f) at (1.5,-3) [rotate=0,spin] {};
	\draw (a) -- node[auto,swap] {$m$} (c);
	\draw (c.35) -- node[auto] {$k$} (b.-145);
	\draw (b) -- node[auto] {$n-k$} (d);
	\draw (c) -- node[auto,swap] {$m$} (e);
	\draw (c.-35) -- node[auto] {$k$} (f.145);
	\draw (f) -- node[auto,swap] {$n-k$} (d);
\end{tikzpicture}
} = \sum_{\ell=0}^k  \frac{\displaystyle {m \brack \ell}\, {k \brack \ell}}{\displaystyle {m+k \brack \ell}}\, 
		\raisebox{-.75cm}{
	\begin{tikzpicture}
		\node(n1) at (0,0) [draw,thick,rectangle,minimum width=1cm] {};
		\node(n2) at (1.5,0) [draw,thick,rectangle,minimum width=1cm] {};
		\node(n1a) at (0,1.5) [draw,thick,rectangle,minimum width=1cm] {};
		\node(n2a) at (1.5,1.5) [draw,thick,rectangle,minimum width=1cm] {};
		\draw (n1.40) to [bend left=90] node[auto,swap] {{$\ell$}} (n2.140);
		\draw (n1a.-40) to [bend right=90] node[auto] {{$\ell$}} (n2a.-140);
		\draw (n1) -- node[auto] {$m-\ell$} (n1a);
		\draw (n2) -- node[auto,swap] {$n-\ell$} (n2a);
	\end{tikzpicture}
	}
\end{equation*}
We see immediately that the change of basis matrix is indeed non-degenerate. The following theorem justifies our change of basis formula.

\end{proof}

\begin{thm}
\label{thm:JWFK}
For any $m,n \in \{0,1,2,\dots \}$, 
\begin{equation*}
\begin{minipage}{1.25cm}
\begin{tikzpicture}[xscale=1,yscale=1,bend angle=90,
spin/.style={rectangle,draw,thick,minimum width=10mm},scale=1]
\node at (0,0) [spin] (n) {};
\draw (n.90) -- node[left] {\small $m+n$} +(0,0.75);
\draw (n.-90) -- node[left] {\small $m+n$} +(0,-0.75);
\end{tikzpicture}
\end{minipage} \quad 
=\ \sum_{k=0}^{\min \{ m,n \}} \frac{\displaystyle {m \brack k}\, {n \brack k}}
{\displaystyle {m+n \brack k}}\quad 
\begin{minipage}{5cm}
\begin{tikzpicture}[xscale=0.75,yscale=0.75,bend angle=90,
spin/.style={rectangle,draw,thick,minimum width=10mm},scale=1]
\node at (-1,1) [spin] (a1) {};
\node at (1,1) [spin] (b1) {};
\node at (-1,-1) [spin] (a2) {};
\node at (1,-1) [spin] (b2) {};
\draw (a1.-150) .. controls +(0,-0.5) and +(0,0.5) .. node[left] {\small $m-k$} (a2.150);
\draw (b1.-30) .. controls +(0,-0.5) and +(0,0.5) .. node[right] {\small $n-k$} (b2.30);
\draw (a1.-30) .. controls +(0,-0.75) and +(0,-0.75) .. node[above] {\small $k$} (b1.-150);
\draw (a2.30) .. controls +(0,0.75) and +(0,0.75) .. node[below] {\small $k$} (b2.150);
\end{tikzpicture}
\end{minipage}
\end{equation*}
\end{thm}

This formula was proved by Frenkel and Khovanov in \cite{frenkel:1997}, where it is Proposition 3.10.
Also see Proposition 2.2 in \cite{kim:2003}
where by unfolding the left and right hand sides via the vector space isomorphism between invariant vectors in $V(m)\otimes V(m) \otimes V(n) \otimes V(n)$ and $\textbf{End}_{\U}(V(m)\otimes V(n))$, 
\begin{equation}
\label{eq:unfolded2}
\begin{minipage}{5.5cm}
\begin{tikzpicture}[xscale=0.75,yscale=0.75,bend angle=90,
spin/.style={rectangle,draw,thick,minimum width=10mm},scale=1]
\node at (-3,0) [spin] (mk) {};
\node at (-1,0) [spin] (ml) {};
\node at (1,0) [spin] (nl) {};
\node at (3,0) [spin] (nk) {};
\node at (0,1.25) [spin] (m) {};
\draw (mk.90) .. controls +(0,0.25) and +(-1,0) .. node[left] {\small $m$} (-1.5,1.75) .. controls +(0.25,0) and +(0,0.25) .. (m.135);
\draw (ml.90) .. controls +(0,0.25) and +(0,-0.25) .. node[left] {\small $m$} (m.-135);
\draw (nl.90) .. controls +(0,0.25) and +(0,-0.25) .. node[right] {\small $n$} (m.-45);
\draw (nk.90) .. controls +(0,0.25) and +(1,0) .. node[right] {\small $n$} (1.5,1.75) .. controls +(-0.25,0) and +(0,0.25) .. (m.45);
\end{tikzpicture}
\end{minipage} \qquad 
=\ \sum_{k=0}^{\min \{m,n \}} c_{m,n,k}\
\begin{minipage}{5cm}
\begin{tikzpicture}[xscale=0.75,yscale=0.75,bend angle=90,
spin/.style={rectangle,draw,thick,minimum width=10mm},scale=1]
\node at (-3,0) [spin] (mk) {};
\node at (-1,0) [spin] (ml) {};
\node at (1,0) [spin] (nl) {};
\node at (3,0) [spin] (nk) {};
\draw (mk.45) .. controls +(0,0.75) and +(0,0.75) .. node[above] {\small $m-k$} (ml.135);
\draw (ml.45) .. controls +(0,0.75) and +(0,0.75) .. node[above] {\small $k$} (nl.135);
\draw (nl.45) .. controls +(0,0.75) and +(0,0.75) .. node[above] {\small $n-k$} (nk.135);
\draw (mk.135) .. controls +(0,2.5) and +(0,2.5) .. node[above] {\small $k$} (nk.45);
\end{tikzpicture}
\end{minipage}\ ,
\end{equation}
$$
c_{m,n,k}\, =\, \frac{\displaystyle {m \brack k}\, {n \brack k}}
{\displaystyle {m+n \brack k}}\ ,
$$

\section{The Graphical Dual Canonical Basis}\label{sec:dual_canonical_basis}

We describe the graphical dual canonical basis first by studying cap diagrams and useful bijections to standard Young tableaux of two or fewer rows and to highest weight vectors in $V(1)\otimes \cdots \otimes V(1)$. One can realize the description here as a special case of the $q$ version of Schur-Weyl duality, where the duality is between the Hecke algebra of type $A_{n_1+\cdots + n_L}$ and $\U$. Then the connection to standard Young tableaux is given by a $q$ version of the Young symmetrizer \cite{jimbo:1986, gyoja:1986}. We remark that the same bijection is considered in \cite{westbury:1995}, but we provide an alternate proof. 

Consider the action of the intertwiner $\delta:V(0) \to V(1) \otimes V(1)$ on the space $V(1)^{\otimes L}$.
First observe that $\delta$ and scalar multiples of $\delta$ are the only elements in $\textbf{Hom}_{\U} (V(0) , V(1) \otimes V(1))$. 
Define a $1$-cap diagram, denoted $C(L,1)$, to be the elements of $\textbf{Hom}_{\U} (V(1)^{\otimes L-2} , V(1)^{\otimes L})$ of the form $\id \otimes \cdots \otimes \id \otimes \delta \otimes \id \otimes \cdots \otimes \id$, where $\delta$ acts on consecutive tensor factors. The set of permissible cap diagrams forms a basis for $\textbf{Hom}_{\U} (V(1)^{\otimes L-2},V(1)^{\otimes L})$. 
Graphically, cap diagrams are all possible non-crossing pairings (drawn as caps) of a linearly ordered finite set. We define a $k$-cap diagram, denoted by $C(L,k)$, to be elements of $\textbf{Hom}_{\U} (V(1)^{\otimes L-2k},V(1)^{\otimes L})$ implemented with $k$ copies of $\delta$ such that the pairings do not cross and $L-2k$ through lines, which also do not cross. When it is clear, these through lines will simply be represented by a dot. $C(L,k)$ forms a basis of $\textbf{Hom}_{\U} (V(1)^{\otimes L-2k},V(1)^{\otimes L})$. 
It is a consequence of the non-crossing property that cap diagrams are uniquely identified by specifying the locations of the right legs (or alternatively, the locations of the left legs).  As usual, this is best illustrated in the following graphical example. 

\begin{example} We list out the cap diagrams in $C(5,2)$ and a bijective correspondence between standard Young tableaux of shape (3,2). The left hand side denotes the locations of positions of right legs.
\begin{align*}
	(2,4) \quad \leftrightarrow \quad
        \begin{tikzpicture}
                \foreach \x in {0,.5,1,1.5,2}{
                \fill [black] (\x,0) circle (1pt);
                }
                \draw[thick] (0,0) arc (180:0:.25);
                \draw[thick] (1,0) arc (180:0:.25);
        \end{tikzpicture}
         & \quad \leftrightarrow \quad
         \raisebox{-.5cm}{
         \begin{tikzpicture}
                 \foreach \x in {0,.5,1}{
                 \draw[thick] (\x,0) +(-.25,-.25) rectangle ++(.25,.25);
                 \draw(0,0) node{1};
                 \draw(0.5,0) node{3};
                 \draw(1,0) node{5};
                 }
                 \foreach \y in {0,.5}{
                 \draw[thick] (\y,-.5) +(-.25,-.25) rectangle ++(.25,.25);
                 \draw(0,-.5) node{2};
                 \draw(.5,-.5) node{4};
                 }
         \end{tikzpicture}
         }\\
	(2,5) \quad \leftrightarrow \quad
                 \begin{tikzpicture}
                \foreach \x in {0,.5,1,1.5,2}{
                \fill [black] (\x,0) circle (1pt);
                }
                \draw[thick] (0,0) arc (180:0:.25);
                \draw[thick] (1.5,0) arc (180:0:.25);
        \end{tikzpicture}
         & \quad \leftrightarrow \quad
         \raisebox{-.5cm}{
         \begin{tikzpicture}
                 \foreach \x in {0,.5,1}{
                 \draw[thick] (\x,0) +(-.25,-.25) rectangle ++(.25,.25);
                 \draw(0,0) node{1};
                 \draw(0.5,0) node{3};
                 \draw(1,0) node{4};
                 }
                 \foreach \y in {0,.5}{
                 \draw[thick] (\y,-.5) +(-.25,-.25) rectangle ++(.25,.25);
                 \draw(0,-.5) node{2};
                 \draw(.5,-.5) node{5};
                 }
         \end{tikzpicture}
         }\\
	(3,4) \quad \leftrightarrow \quad
                 \begin{tikzpicture}
                \foreach \x in {0,.5,1,1.5,2}{
                \fill [black] (\x,0) circle (1pt);
                }
                \draw[thick] (0.5,0) arc (180:0:.25);
                \draw[thick] (0,0) arc (180:0:.75);
        \end{tikzpicture}
         & \quad \leftrightarrow \quad
         \raisebox{-.5cm}{
         \begin{tikzpicture}
                 \foreach \x in {0,.5,1}{
                 \draw[thick] (\x,0) +(-.25,-.25) rectangle ++(.25,.25);
                 \draw(0,0) node{1};
                 \draw(0.5,0) node{2};
                 \draw(1,0) node{5};
                 }
                 \foreach \y in {0,.5}{
                 \draw[thick] (\y,-.5) +(-.25,-.25) rectangle ++(.25,.25);
                 \draw(0,-.5) node{3};
                 \draw(.5,-.5) node{4};
                 }
         \end{tikzpicture}
         }\\
	(3,5) \quad \leftrightarrow \quad
                 \begin{tikzpicture}
                \foreach \x in {0,.5,1,1.5,2}{
                \fill [black] (\x,0) circle (1pt);
                }
                \draw[thick] (0.5,0) arc (180:0:.25);
                \draw[thick] (1.5,0) arc (180:0:.25);
        \end{tikzpicture}
         & \quad \leftrightarrow \quad
         \raisebox{-.5cm}{
         \begin{tikzpicture}
                 \foreach \x in {0,.5,1}{
                 \draw[thick] (\x,0) +(-.25,-.25) rectangle ++(.25,.25);
                 \draw(0,0) node{1};
                 \draw(0.5,0) node{2};
                 \draw(1,0) node{4};
                 }
                 \foreach \y in {0,.5}{
                 \draw[thick] (\y,-.5) +(-.25,-.25) rectangle ++(.25,.25);
                 \draw(0,-.5) node{3};
                 \draw(.5,-.5) node{5};
                 }
         \end{tikzpicture}
         }\\
	(4,5) \quad \leftrightarrow \quad
                 \begin{tikzpicture}
                \foreach \x in {0,.5,1,1.5,2}{
                \fill [black] (\x,0) circle (1pt);
                }
                \draw[thick] (0.5,0) arc (180:0:.75);
                \draw[thick] (1,0) arc (180:0:.25);
        \end{tikzpicture}
         & \quad \leftrightarrow \quad
         \raisebox{-.5cm}{
         \begin{tikzpicture}
                 \foreach \x in {0,.5,1}{
                 \draw[thick] (\x,0) +(-.25,-.25) rectangle ++(.25,.25);
                 \draw(0,0) node{1};
                 \draw(0.5,0) node{2};
                 \draw(1,0) node{3};
                 }
                 \foreach \y in {0,.5}{
                 \draw[thick] (\y,-.5) +(-.25,-.25) rectangle ++(.25,.25);
                 \draw(0,-.5) node{4};
                 \draw(.5,-.5) node{5};
                 }
         \end{tikzpicture}
         }
\end{align*}
\end{example}
We now describe the structure of cap diagrams by giving a bijection to standard Young tableaux.
\begin{prop}
	There exist bijections from $C(L,k)$ to standard Young tableaux of shape $(L-k,k)$ and to the highest weight vectors of $V(1)^{\otimes L}$ of weight $q^{L-2k}$. 	
\label{prop:bijection}
\end{prop}
\begin{proof}
	We first remark that elements of $\textbf{Hom}_{\U}(V(1)^{\otimes L-2k},V(1)^{\otimes L})$ send highest weight vectors to highest weight vectors or to the zero vector. We look at the non-zero images of the highest weight vector with weight $q^{L-2k}$, and observe that in the Temperley-Lieb basis, the elements in the image correspond precisely to the possible cap diagrams, and also to all possible highest weight vectors with weight $q^{L-2k}$. This is illustrated as follows:
	\begin{equation*}
		\raisebox{-.25cm}{
		\begin{tikzpicture}
                \foreach \x in {0,.5,1,1.5,2,2.5,3}{
                \fill [black] (\x,0) circle (1pt);
                }
                \foreach \x in {1,1.5,2}{
                \fill [black] (\x,1) circle (1pt);
		\draw[thick,->] (\x,1) -- ++(0,.5);
		}
                \draw[thick] (1,0) arc (180:0:.25);
                \draw[thick] (2,0) arc (180:0:.25);
		\draw[thick] (1,1) .. controls (1,.75) and (0,.25) .. (0,0);	
		\draw[thick] (1.5,1) .. controls (1.5,.75) and (0.5,.25) ..(0.5,0);	
		\draw[thick] (2,1) .. controls (2,.75) and (3,.25) .. (3,0);	
		\end{tikzpicture} 
		}= 
		\raisebox{-.25cm}{
		\begin{tikzpicture}
                \foreach \x in {0,.5,1,1.5,2,2.5,3}{
                \fill [black] (\x,0) circle (1pt);
                }
                \draw[thick] (1,0) arc (180:0:.25);
                \draw[thick] (2,0) arc (180:0:.25);
                \foreach \x in {0,.5,3}{
		\draw[thick,->] (\x,0) -- ++(0,.5);
		}
		\end{tikzpicture}
		}
	\end{equation*}
	Thus, by counting the dimension of the highest weight vectors of weight $q^{L-2k}$ in $V(1)^{\otimes L}$, we count the number of cap diagrams. The dimension of the space of heighest weights of weight $q^{L-2k}$ is $\binom L k - \binom L {k-1}$, since in the standard tensor basis, there are $\binom L k$ vectors with weight $q^{L-2k}$. To ensure that we only count highest weight vectors, we subtract out all of those vectors that are obtained by applying the lowering operator to a vector of weight $q^{L-2(k-1)}$, of which there are $\binom L {k-1}$. 
	Concerning the number of standard Young tableaux of shape $(L-k,k)$, it is an easy consequence of the hook length formula that there are precisely $\binom L k -\binom L {k-1}$. Moreover, the locations of the right legs of elements in $C(L,k)$ correspond exactly to the entries of the second row of the tableaux, and so we have the desired bijection.
\end{proof}
Now let $\beta^{L-2k}$ be the elements of the standard tensor basis of $V(1)^{\otimes L-2k}$ in which all $\down$ occur to the left of $\up$. We define the \emph{dual canonical basis}, denoted by $\tilde{\beta}^{L-2k}$, to be the union of the images of $\beta^{L-2k}$ mapped through $C(L,k)$. In symbols,
\begin{equation*}
	\tilde{\beta}^{L-2k} = \bigcup_{d \in C(L,k)} \textrm{Im } d(\beta^{L-2k})
\end{equation*}
For higher dimensional irreducible representations, $V(n_1)\otimes \cdots \otimes V(n_L)$, we simply consider the non-zero elements of $\pi_{n_1} \otimes \cdots \otimes \pi_{n_L} (\tilde{\beta}^{n_1+ \cdots + n_L -2k})$ to be the dual canonical basis. 
\subsection{$\U$ Action on the Dual Canonical Basis}
We briefly recall from \cite{frenkel:1997} the action of $\U$ on the dual canonical basis.
                      The action of the lowering operator, $F$, on a basis vector $v$ is as follows:
                      \begin{enumerate}
                              \item Consider the $i$th $\up$ from the right. Flip it so that it becomes $\down$.
                              \item If there is an $\up$ to the left of this flipped arrow, contract the two into an arc.
                              \item Multiply the resulting vector by $[i]$
                              \item Repeat for all $\up$ arrows. Sum the results.
                      \end{enumerate}
\begin{example}
	        \begin{align*}
			F &
			\raisebox{-1cm}{
                      \begin{tikzpicture}[scale=.8]
                              \node (n1) at (0,0) [draw, thick, rectangle,minimum width=1cm,minimum
			      height=.5cm] {\tiny{$3$}};
                              \node (n2) at (1.5,0) [draw, thick, rectangle,minimum width=1cm,minimum
			      height=.5cm] {\tiny{$3$}};
                              \node (n3) at (3,0) [draw, thick, rectangle,minimum width=1cm,minimum
			      height=.5cm] {\tiny{$3$}};
                              \node (n4) at (4.5,0) [draw, thick, rectangle,minimum width=1cm,minimum
			      height=.5cm] {\tiny{$3$}};
                              \node (n5) at (6,0) [draw, thick, rectangle,minimum width=1cm,minimum
			      height=.5cm] {\tiny{$3$}};
                              \draw (n2.40) to [bend left=90] (n3.140);
                              \draw (n3.40) to [bend left=90] (n4.140);
                              \draw (n3.90) to [bend left=90] (n4.90);
                              \draw (n2.90) to [bend left=90] (n4.40);
                              \draw[ <-] (n1.140) -- ++(0,.5);
                              \draw[ <-] (n1.90) -- ++(0,.5);
                              \draw[ ->] (n1.40) -- ++(0,.5);
                              \draw[ ->] (n2.140) -- ++(0,.5);
                              \draw[ ->] (n5.140) -- ++(0,.5);
                              \draw[ ->] (n5.90) -- ++(0,.5);
                              \draw[ ->] (n5.40) -- ++(0,.5);
                      \end{tikzpicture}
		      } = 
			  [3]
                \raisebox{-1cm}
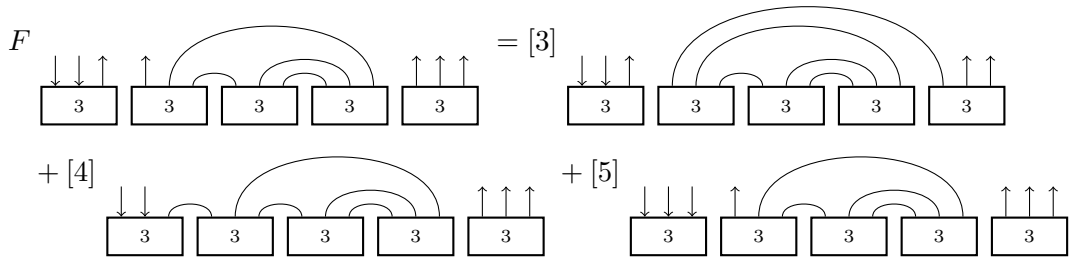
{
                      \begin{tikzpicture}[scale=.8]
                              \node (n1) at (0,0) [draw, thick, rectangle,minimum width=1cm,minimum height=.5cm] {\tiny{$3$}};
                              \node (n2) at (1.5,0) [draw, thick, rectangle,minimum width=1cm,minimum height=.5cm] {\tiny{$3$}};
                              \node (n3) at (3,0) [draw, thick, rectangle,minimum width=1cm,minimum height=.5cm] {\tiny{$3$}};
                              \node (n4) at (4.5,0) [draw, thick, rectangle,minimum width=1cm,minimum height=.5cm] {\tiny{$3$}};
                              \node (n5) at (6,0) [draw, thick, rectangle,minimum width=1cm,minimum height=.5cm] {\tiny{$3$}};
                              \draw (n2.40) to [bend left=90] (n3.140);
                              \draw (n3.40) to [bend left=90] (n4.140);
                              \draw (n3.90) to [bend left=90] (n4.90);
                              \draw (n2.90) to [bend left=90] (n4.40);
                              \draw (n2.140) to [bend left=90] (n5.140);
                              \draw[ <-] (n1.140) -- ++(0,.5);
                              \draw[ <-] (n1.90) -- ++(0,.5);
                              \draw[ ->] (n1.40) -- ++(0,.5);
                              \draw[ ->] (n5.90) -- ++(0,.5);
                              \draw[ ->] (n5.40) -- ++(0,.5);
                      \end{tikzpicture}
                      }
		      \\
                      & + [4]
                \raisebox{-1cm}{
                      \begin{tikzpicture}[scale=.8]
                              \node (n1) at (0,0) [draw, thick, rectangle,minimum width=1cm,minimum height=.5cm] {\tiny{$3$}};
                              \node (n2) at (1.5,0) [draw, thick, rectangle,minimum width=1cm,minimum height=.5cm] {\tiny{$3$}};
                              \node (n3) at (3,0) [draw, thick, rectangle,minimum width=1cm,minimum height=.5cm] {\tiny{$3$}};
                              \node (n4) at (4.5,0) [draw, thick, rectangle,minimum width=1cm,minimum height=.5cm] {\tiny{$3$}};
                              \node (n5) at (6,0) [draw, thick, rectangle,minimum width=1cm,minimum height=.5cm] {\tiny{$3$}};
                              \draw (n2.40) to [bend left=90] (n3.140);
                              \draw (n3.40) to [bend left=90] (n4.140);
                              \draw (n3.90) to [bend left=90] (n4.90);
                              \draw (n2.90) to [bend left=90] (n4.40);
                              \draw (n1.40) to [bend left=90] (n2.140);
                              \draw[ <-] (n1.140) -- ++(0,.5);
                              \draw[ <-] (n1.90) -- ++(0,.5);
                              \draw[ ->] (n5.140) -- ++(0,.5);
                              \draw[ ->] (n5.90) -- ++(0,.5);
                              \draw[ ->] (n5.40) -- ++(0,.5);
                      \end{tikzpicture}
                      }
                       + [5]
                \raisebox{-1cm}{
                      \begin{tikzpicture}[scale=.8]
                              \node (n1) at (0,0) [draw, thick, rectangle,minimum width=1cm,minimum height=.5cm] {\tiny{$3$}};
                              \node (n2) at (1.5,0) [draw, thick, rectangle,minimum width=1cm,minimum height=.5cm] {\tiny{$3$}};
                              \node (n3) at (3,0) [draw, thick, rectangle,minimum width=1cm,minimum height=.5cm] {\tiny{$3$}};
                              \node (n4) at (4.5,0) [draw, thick, rectangle,minimum width=1cm,minimum height=.5cm] {\tiny{$3$}};
                              \node (n5) at (6,0) [draw, thick, rectangle,minimum width=1cm,minimum height=.5cm] {\tiny{$3$}};
                              \draw (n2.40) to [bend left=90] (n3.140);
                              \draw (n3.40) to [bend left=90] (n4.140);
                              \draw (n3.90) to [bend left=90] (n4.90);
                              \draw (n2.90) to [bend left=90] (n4.40);
                              \draw[ <-] (n1.140) -- ++(0,.5);
                              \draw[ <-] (n1.90) -- ++(0,.5);
                              \draw[ <-] (n1.40) -- ++(0,.5);
                              \draw[ ->] (n2.140) -- ++(0,.5);
                              \draw[ ->] (n5.140) -- ++(0,.5);
                              \draw[ ->] (n5.90) -- ++(0,.5);
                              \draw[ ->] (n5.40) -- ++(0,.5);
                      \end{tikzpicture}
                      }
        \end{align*}
	Observe that the terms with both ends of an arc attached to the same projector have been evaluated as the zero vector.
\end{example}
If $u$ is the involution that reflects about a vertical axis and changes up to down and down to up, then $E= u \circ F \circ u$.

Observe that particularly in the ``generic'' case where we specialize to $0<q<1$, that the generators of $\U$, $E$, $F$, and $K$, act on the dual canonical basis as non-negative matrices. We shall see this property play a crucial role in conjunction with the Perron-Frobenius theorem.
\subsection{Action of the Cascade Basis on the Dual Canonical Basis}
We now wish to calculate the action of the Cascade basis on the Dual Canonical basis of $V(n_1) \otimes \cdots \otimes V(n_L)$. In particular, we will prove the following proposition.

\begin{prop}
The cascade basis has all positive matrix entries in the dual canonical basis.
\label{prop:cascadepositivity}
\end{prop}
 \begin{proof}
 Consider the following simplifications of the diagram which are consequences of Proposition \ref{prop:FK}
 \begin{align*}
 \raisebox{-2cm}{
 \begin{tikzpicture}
 [symm/.style={draw,thick, rectangle,minimum width=1cm}]
 \node (n1) at (0,0) [symm] {};
  \node (n2) at (1.5,0) [symm] {};
   \node (n1a) at (0,3) [symm] {};
  \node (n2a) at (1.5,3) [symm] {};
   \node (n1aa) at (0,4) [symm] {};
  \node (n2aa) at (1.5,4) [symm] {};
  \node[draw,thick,rectangle,minimum width=1.5cm] (mid) at (.25,1.5) {};
  \draw (n1.90) to node[auto] {$m$} (mid.-140);
  \draw (n1a.-90) to node[auto,swap] {$m$} (mid.140);
  \draw (n2) to node[auto,swap] {$n-k$}(n2a);
  \draw (mid.40) to node[auto] {$k$} (n2a.-130);
  \draw (mid.-40) to node[auto] {$k$} (n2.130);
  \draw (n1a.50) to [bend left=90] node[auto] {$j$} (n2a.130);
  \draw (n1a.90) -- node[auto] {$m-j$} (n1aa);
  \draw (n2a.90) -- node[auto,swap] {$n-j$} (n2aa);
  \end{tikzpicture}
  }
  & = 
 \raisebox{-2cm}{
 \begin{tikzpicture}
 [symm/.style={draw,thick, rectangle, minimum width=1cm}]
 \node (n1) at (0,0) [symm] {};
  \node (n2) at (1.5,0) [symm] {};
  \node (n2a) at (1.5,3) [symm] {};
  \node (n1aa) at (0,4) [symm] {};
  \node (n2aa) at (1.5,4) [symm] {};
  \node[draw,thick,rectangle, minimum width=1.5cm] (mid) at (.25,1.5) {};
  \draw (n1.90) to node[auto] {$m$} (mid.-140);
  \draw (mid.90) -- ++(0,1.5) to [bend left=90] node[auto] {$j$} (n2a.130);
  \draw (n2) to node[auto,swap] {$n-k$} (n2a);
  \draw (mid.40) to node[auto] {$k$} (n2a.-130);
  \draw (mid.-40) to node[auto] {$k$} (n2.130);
  \draw (mid.140) to node[auto] {$m-j$} (n1aa);
  \draw (n2a) to node[auto,swap] {$n-j$} (n2aa);
  \end{tikzpicture}
  }  \\
&= \frac{[n-j]![n-k]!}{[n-j-k]![n]!}
 \raisebox{-2cm}{
 \begin{tikzpicture}
 [symm/.style={draw,thick, rectangle, minimum width=1cm}]
 \node (n1) at (0,0) [symm] {};
  \node (n2) at (1.5,0) [symm] {};
  \node (n2a) at (1.5,3) [symm] {};
  \node (n1aa) at (0,4) [symm] {};
  \node (n2aa) at (1.5,4) [symm] {};
  \node[draw,thick,rectangle, minimum width=1.5cm] (mid) at (.25,1.5) {};  
  \draw (n1.90) to node[auto] {$m$} (mid.-140);  
\draw (n2.80) -- ++(0,1.5) to [bend right=90] node[auto] {$j$} (mid.15);  
\draw (n2.45) to node[auto,swap] {$n-j-k$} (n2a.-45);  
\draw (mid.90) to node[auto] {$k$} (n2a.-130);  
\draw (mid.-40) to node[auto] {$k$} (n2.130);  
  \draw (mid.140) to node[auto] {$m-j$} (n1aa);
  \draw (n2a) to node[auto,swap] {$n-j$} (n2aa);
\end{tikzpicture}  
} \\
& = \frac{[m+k-j]![m]![n-j]![n-k]!}{[m-j]![m+k]![n-j-k]![n]!}
 \raisebox{-2cm}{
 \begin{tikzpicture}
 [symm/.style={draw,thick, rectangle,minimum width=1cm}]
 \node (n1) at (0,0) [symm] {};
  \node (n2) at (1.5,0) [symm] {};
  \node (n2a) at (1.5,3) [symm] {};
  \node (n1aa) at (0,4) [symm] {};
  \node (n2aa) at (1.5,4) [symm] {};
  \node[draw,thick,rectangle,minimum width=1.5cm] (mid) at (.25,1.5) {};  
  \draw (n1.90) to node[auto] {$m-j$} (mid.-140);  
\draw (n1.20) to [bend left=90] node[auto,swap] {$j$} (n2.160);
\draw (n2.45) to node[auto,swap] {$n-j-k$} (n2a.-45);  
\draw (mid.90) to node[auto] {$k$} (n2a.-130);  
\draw (mid.-40) to node[auto] {$k$} (n2.80);  
  \draw (mid.140) to node[auto] {$m-j$} (n1aa);
  \draw (n2a) to node[auto,swap] {$n-j$} (n2aa);
\end{tikzpicture}
} \\
& = \sum_{\ell=0}^k c_\ell
\raisebox{-1.25cm}{
\begin{tikzpicture}
 [symm/.style={draw,thick, rectangle,minimum width=1cm}]
 \node (n1) at (0,0) [symm] {};
  \node (n2) at (1.5,0) [symm] {};
  \node (n2a) at (1.5,2) [symm] {};
  \node (n1a) at (0,2) [symm] {};
\draw (n1.20) to [bend left=90] node[auto] {$j+\ell$} (n2.160);
\draw (n1a.-20) to [bend right=90] node[auto,swap] {$\ell$} (n2a.-160);
\draw (n1) to node[auto] {$m-j-\ell$} (n1a);
\draw (n2) to node[auto,swap] {$n-j-\ell$} (n2a);
\end{tikzpicture}
}
  \end{align*}
 where 
 $$
 c_\ell\, =\, \frac{[m]![k]![n-j]![n-k]![m+k-j-\ell]!}{[n]![\ell]![m+k]![k-\ell]![n-j-k]![m-j-\ell]!}\, .
 $$ 
 The last equality follows by using Theorem \ref{thm:JWFK} on the projector of size $m+k-j$. 

The positivity result now follows from the fact that the two projectors of size $m-j$ and $n-j$ may be decomposed into a positive sum of elements of the Temperley-Lieb Algebra. All non-zero terms of this expansion have been reduced to elements of the dual canonical basis, and by our calculations, they have positive coefficients. 

\end{proof}
\begin{remark}
	The name Cascade basis is inspired by the above calculation in which the down-turned arcs ``cascade'' downwards as the calculation proceeds.
\end{remark}

\section{Sufficient Conditions for FOEL in the Dual Canonical Basis}
\label{sec:sufficient_conditions}

\subsection{Strictly Positive Eigenvectors of Non-negative Matrices}

A non-negative $n\times n$ matrix $A$ is said to be irreducible if for every $i,j \in \{1,\dots,n\}$ there exists a positive integer $m$ such that $(A^m)_{ij} > 0$.
\begin{thm}[Perron-Frobenius]
\label{thm:PF}
	Let $A$ be a non-negative, irreducible $n \times n$ matrix, and let $\rho(A)$ denote the spectral radius of $A$. Then there exists a unique eigenvector
$\boldsymbol{x} = (x_1,\dots,x_n)^T$ such that
\begin{itemize}
\item $x_1,\dots,x_n > 0$,
\item $x_1+\dots+x_n=1$,
\item $A \boldsymbol{x} = \rho(A) \boldsymbol{x}$
\end{itemize}
\end{thm}
\begin{proof}
	See \cite{dym:2007}
\end{proof}

The following lemma gives a useful condition for distinguishing the eigenvector with maximal eigenvalue.

\begin{lemma}
Suppose that $A$ is an $n\times n$ non-negative matrix.  
Suppose that $\boldsymbol{x} = (x_1,\dots,x_n)^T$ is a vector such that
$x_1,\dots,x_n>0$, and suppose that $A \boldsymbol{x} = \lambda \boldsymbol{x}$
for some scalar $\lambda$. Then $\lambda = \rho(A)$.
\end{lemma}

\begin{proof}
Let us define a ray of matrices, $B(t)$ for $t\geq 0$ such that
$$
b_{ij}(t) = a_{ij} + t x_i\, .
$$
Then $B(0) = A$ and $B(t)$ is strictly positive for each $t>0$.
Moreover, for the positive eigenvector $\boldsymbol{x}$ of the matrix $A$, we have
$$
(B \boldsymbol{x})_i\, =\, (A \boldsymbol{x})_i + t \sum_{j=1}^{n} x_i x_j\, =\, (\lambda + t \sum_{j=1}^{n} x_j) x_i\, .
$$
So $\boldsymbol{x}$ is a positive eigenvector of $B(t)$.
By Theorem \ref{thm:PF}, this means that the spectral radius $\rho(B(t))$ is given by
the eigenvalue:
$$
\rho(B(t))\, =\, \lambda + t \sum_{j=1}^{n} x_j\, .
$$
The eigenvectors and eigenvalues of $B(t)$ are continuous functions
of $t$, although not necessarily differentiable functions.
Therefore, taking the limit, we see that
$$
\rho(A)\, =\, \lim_{t \downarrow 0} \rho(B(t))\, =\, \lambda\, .
$$
\end{proof}

\begin{prop}
Suppose that $A$ is an $n\times n$ real matrix such that $a_{ij} \geq 0$ for all $i,j \in \{1,\dots,n\}$ such that $i\neq j$, and that for $t$ sufficiently large, $A+tI$ is non-negative. 
Suppose that $\boldsymbol{x} = (x_1,\dots,x_n)^T$ is a vector such that
$x_1,\dots,x_n>0$, and suppose that $A \boldsymbol{x} = \lambda \boldsymbol{x}$
for some scalar $\lambda$.
Then
$$
\lambda\, =\, \max\{ \textrm{Re}(\mu)\, :\, \mu \in \operatorname{spec}(A) \}\, .
$$
\label{prop:posEigen}
\end{prop}

\begin{proof}
Apply the previous lemma to the matrices $A+tI$ for $t>0$.
Since
$$
\operatorname{spec}(A+tI)\, =\, \{\mu + t \, :\, \mu \in \operatorname{spec}(A)\}\, ,
$$
we see that
\begin{align*}
\rho(A+tI)\, 
&=\, \max_{\mu \in \operatorname{spec}(A)} |\mu + t|\\
&=\, \max_{\mu \in \operatorname{spec}(A)} \sqrt{[\textrm{Re}(\mu) + t]^2 + [\textrm{Im}(\mu)]^2}\\
\end{align*}
For $t$ sufficiently large the eigenvalue $\mu$ which attains the maximum
will have the maximal real part.
But also for $t$ sufficiently large, $A+tI$ will have nonnegative entries, because
$t$ will be larger than the largest negative part of any diagonal matrix entry of $A$.
Therefore, we will have that $\rho(A+tI) = \lambda + t$.
From this we conclude the result.
\end{proof}

\subsection{Irreducibility in Highest Weight Spaces}
Let $A$ be an $n \times n$ matrix.  For any two basis vectors $v_i$ and $v_j$, if there exist $n_i, n_j \in \nats$ such that $(A^{n_i})_{ji} \neq 0$ and $(A^{n_j})_{ij} \neq 0$, we say that $v_i$ is connected to $v_j$ with respect to $A$. This is an equivalence relation, in particular it is symmetric. If every basis vector is connected to all others, then $A$ is irreducible. 

Let $\Hi^{(n_1 + \cdots + n_L -2k)}_{HW}$ denote the space of highest weight vectors with weight $q^{n_1+ \cdots + n_L -2k}$. 

Recall that each dual canonical basis vector of $\Hi^{(n_1+\cdots + n_L - 2k)}_{HW}$ is uniquely characterized by a cap diagram with $k$ arcs. Such a diagram in turn is also uniquely characterized by the positions of right legs of arcs. Number the arcs from left to right according to position of right legs. We record the positions of the right legs in a $k$-tuple $(x_1,x_2, \ldots, x_k)$ where $x_i \in \{1,2, \ldots,L\}$. 

Thus, given a dual canonical basis vector, we can define its \emph{height}.
\begin{equation*}
	h(v) = \sum_{i=1}^k x_i
\end{equation*}
We say that an arc at position $x_i$ is in \emph{maximal position} if $x_i +1$ substituted for $x_i$ is not a permissible configuration of arcs.
There is a unique highest vector with all arcs in maximal position. 

We wish to demonstrate that $H= \sum_{i=1}^{L-1} \sum_{k^\prime =1}^{\min (n_i,n_{i+1})} J_{k^\prime}^{(i)} K_{n_i,n_{i+1}}(k^\prime)$ is irreducible in the space $\Hi^{(n_1 + \cdots + n_L-2k)}_{HW}$ in the dual canonical basis for any $k \in \{0, 1, \ldots , \left\lfloor \frac {n_1 + \cdots + n_L}2 \right\rfloor \}$ where we assume that for each $i$, there is at least one $k^\prime$ such that $J_{k^\prime}^{(i)}\neq 0$. It will suffice to show the following proposition.

\begin{prop}
	Every vector in the dual canonical basis of $\Hi^{(n_1+\cdots + n_L - 2k)}_{HW}$ is connected to the highest vector with respect to $H$.
	\label{prop:irreducible}
\end{prop}

\begin{proof}
	Let $v$ be an element of the dual canonical basis in the space $\Hi^{(n_1+\cdots+n_L-2k)}_{HW}$ with the right legs of arcs located at sites $(x_1,x_2, \ldots, x_k)$. Suppose that $\star$ is an arc that is not in maximal position, located at site $i$. By interaction between sites $i$ and $i+1$ via the $\ell=1$ term in the calculation in Proposition \ref{prop:cascadepositivity}, we see that there is a coefficient of the ensuing expansion with $\star$ replaced by an arc at site $i+1$.
This will always be possible with the rightmost arc, $\star$, that is not in maximal position, since by non-maximality, we can choose expansions of the two upper projectors that connect $\star$ to a free arrow or to the left leg of an arc. For example,

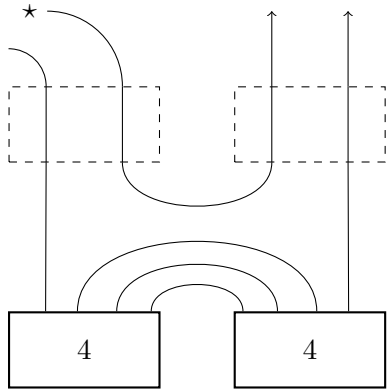
\begin{equation*}
	\begin{tikzpicture}
		\node (n1) at (0,0) [draw,thick,rectangle,minimum height=1cm,minimum width=2cm] {$4$};
		\node (n2) at (3,0) [draw,thick,rectangle,minimum height=1cm,minimum width=2cm] {$4$};
		\node (n1a) at (0,3) [draw,dashed,rectangle,minimum height=1cm,minimum width=2cm] {};
		\node (n2a) at (3,3) [draw,dashed,rectangle,minimum height=1cm,minimum width=2cm] {};
		\draw[] (n1.135) to (n1a.-135);
		\draw[] (n1a.45) to (n1a.-45);
		\draw[] (n2.45) to (n2a.-45);
		\draw[] (n1.30) to [bend left=90] (n2.150);
		\draw[] (n1.50) to [bend left=90] (n2.130);
		\draw[] (n1.100) to [bend left=90] (n2.80);
		\draw[] (n1a.-45) to [bend right=90] (n2a.-135);
		\draw[] (n1a.135) arc (0:90:.5);
		\draw[] (n1a.45) arc (0:90:1) node[anchor=east] {$\star$};
		\draw[->] (n2a.45) -- ++(0,1);
		\draw[->] (n2a.135) -- ++(0,1);
		\draw[] (n1a.-135) to (n1a.135);
		\draw[] (n2a.-45) -- (n2a.45);
		\draw[] (n2a.-135) -- (n2a.135);
	\end{tikzpicture}
\end{equation*}
The end result is that the arc with right leg labelled by $\star$ at site $x_i$ is replaced by a new arc with right leg at site $x_i+1$. 
Thus by repeated application of $H$, we traverse from $v$ to the highest vector by raising the rightmost non-maximal arc until we can do so no longer.
Now by reflecting about a vertical axis, one sees that it is possible to traverse from the highest vector to the lowest vector by repeating a similar argument for left legs of arcs. Finally, from the lowest vector, we traverse back to $v$ by setting arc $x_k$ in place, then $x_{k-1}$, and so on until we have obtained $v$. Thus any $v$ is connected to the highest vector, and thus $H|_{\Hi^{(L-2k)}_{HW}}$ is irreducible. 
\end{proof}
Recall that we are interested in the class of Hamiltonians
\begin{equation}
	H= \sum_{i=1}^{L-1} \sum_{k^\prime =1}^{\min (n_i,n_{i+1})} J_{k^\prime}^{(i)} K_{n_i,n_{i+1}}(k^\prime)
	\label{eqn:HamiltonianClass}
\end{equation}
where $J_{k^\prime}^{(i)} \leq 0 $
Thus when we restrict to any particular $\Hi_{HW}^{(n_1+\cdots+n_L-2k)}$, $t-H$ is a non-negative irreducible matrix in the dual canonical basis. By applying the Perron-Frobenius theorem to $t-H$, one obtains a strictly positive (highest weight) eigenvector of $t-H$ (or equivalently of $H$) which we denote by $\phi^{(k)}$. This eigenvector has eigenvalue $t-E_0(H,n_1 + \cdots + n_L -2k)$. 
\subsection{Comparing Energy Levels}
We are now able to provide sufficient conditions for FOEL for Hamiltonians in the form of equation \ref{eqn:HamiltonianClass}.
\begin{thm}
	Let $H$ be given such that there is some $k^\prime$ such that $J_{k^\prime}^{(i)} \neq 0$ for each $i$, and such that all off-diagonal matrix elements of $H$ in the dual canonical basis are non-positive. Then
	\begin{equation*}
		E_0(H,n_1 + \cdots + n_L -2k) \leq E_0(H,n_1 + \cdots n_L - 2k - 2) \textrm{ for all } 0 < k < \left\lfloor \frac {n_1+\cdots+n_L}2 \right\rfloor -1
	\end{equation*}
or equivalently, FOEL is satisfied.
\end{thm}
\begin{proof}
	First, we remark that the condition that $J_{k^\prime}^{(i)}$ be non-zero is only required to ensure that $H$ is irreducible so that the Perron-Frobenius theorem applies. 

	We begin by demonstrating that $H$ is lower block triangular in $F\Hi^{(n_1 +\cdots + n_L-2k)}_{HW}$. By the action of the lowering operator, one can see that 
	\begin{equation}
		\begin{split}
		F\Hi^{(n_1 + \cdots + n_L -2k)}_{HW} & = \textrm{span} \{ \textrm{ one } \down , k \textrm{ arcs }\} \oplus \textrm{span} \{ \textrm{ zero } \down, k+1 \textrm{ arcs } \} \\ 
		& = \textrm{span} \{ \textrm{ one } \down, k \textrm{ arcs } \} \oplus \Hi^{(n_1+ \cdots+ n_L -2k-2)}_{HW}
		\end{split}
		\label{eqn:FHdecomp}
	\end{equation}
	Additionally, we have the fact that in the dual canonical basis, $H$ never decreases the number of arcs since the local Hamiltonians expand via the calculation in Proposition \ref{prop:cascadepositivity}, and the upturned $\ell$ arcs cannot form self-loops. In the space $F\Hi^{(n_1 + \cdots + n_L -2k)}_{HW}$, it is possible for the action of $H$ to add one arc at the expense of removing a $\down$. 
Consider 
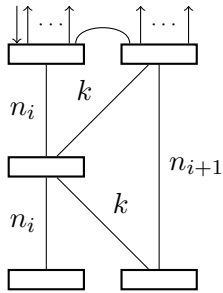
\begin{equation*}
	\begin{tikzpicture}
		\node (n1) at (0,0) [draw,thick,rectangle, minimum width=1cm] {};
		\node (n2) at (1.5,0) [draw,thick,rectangle, minimum width=1cm] {};
		\node (n1a) at (0,3) [draw,thick,rectangle, minimum width=1cm] {};
		\node (n2a) at (1.5,3) [draw,thick,rectangle, minimum width=1cm] {};
		\node (mid) at (0,1.5) [draw,thick,rectangle,minimum width=1cm] {};
		\draw[] (n1) to node[auto] {$n_i$} (mid);
		\draw[] (mid) to node[auto] {$n_i$} (n1a);
		\draw[] (mid.45) to node[auto] {$k$} (n2a);
		\draw[] (mid.-45) to node[auto] {$k$} (n2);
		\draw[] (n1a.20) to [bend left=90] (n2a.160);
		\draw (n2) to node[auto,swap] {$n_{i+1}$} (n2a);
		\draw[<-] (n1a.160) -- ++(0,.5);
		\draw[->] (n1a.150) -- node[right] {\tiny{$\cdots$}} ++(0,.5);
		\draw[->] (n1a.25) -- ++(0,.5);
		\draw[->] (n2a.150) -- node[right] {\tiny{$\cdots$}} ++(0,.5);
		\draw[->] (n2a.20) -- ++(0,.5);
	\end{tikzpicture}
\end{equation*}
Again by the calculation in Proposition \ref{prop:cascadepositivity}, we will obtain a term of the form

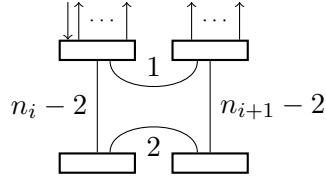
\begin{equation*}
	\begin{tikzpicture}
		\node (n1) at (0,0) [draw,thick,rectangle, minimum width=1cm] {};
		\node (n2) at (1.5,0) [draw,thick,rectangle, minimum width=1cm] {};
		\node (n1a) at (0,1.5) [draw,thick,rectangle, minimum width=1cm] {};
		\node (n2a) at (1.5,1.5) [draw,thick,rectangle, minimum width=1cm] {};
		\draw[] (n1) to node[auto] {$n_i-2$} (n1a);
		\draw[] (n2) to node[auto,swap] {$n_{i+1}-2$} (n2a);
		\draw[] (n1.40) to [bend left=90] node[auto,swap] {$2$} (n2.140);
		\draw[] (n1a.-40) to [bend right=90] node[auto] {$1$} (n2a.-140);
		\draw[<-] (n1a.160) -- ++(0,.5);
		\draw[->] (n1a.150) -- node[right] {\tiny{$\cdots$}} ++(0,.5);
		\draw[->] (n1a.20) -- ++(0,.5);
		\draw[->] (n2a.150) -- node[right] {\tiny{$\cdots$}} ++(0,.5);
		\draw[->] (n2a.20) -- ++(0,.5);
	\end{tikzpicture}
\end{equation*}
Expanding out the top projectors as $q$-symmetrizers, we see that the only terms that survive
are the ones in which the $\down$ is assigned to the left end of the upturned arc. By similar
reasoning, it is easy to see that $H$ acting on a highest weight space will never introduce a 
$\down$. Hence we have proved the following statement:
	$H$ is lower block triangular with 
\begin{equation*}
		H \left( \textrm{span} \{ one \down , k \textrm{ arcs } \} \right) \subseteq \textrm{span} \{ one \down , k \textrm{ arcs } \} \oplus \Hi_{HW}^{(n_1+\cdots+n_L-2k-2)}
\end{equation*}
	and 
\begin{equation*}
		H \left( \Hi_{HW}^{(n_1 + \cdots + n_L - 2k-2) } \right) \subseteq \Hi_{HW}^{(n_1 + \cdots + n_L - 2k-2) } 
\end{equation*}

Let $\varphi^{(n_1+\cdots+n_L - 2k)}$ be the Perron-Frobenius eigenvector for the matrix $t-H|_{\Hi_{HW}^{(n_1+ \cdots + n_L - 2k)}}$ in the dual canonical basis, and let $\varphi^{(n_1+\cdots + n_L -2k-2)}$ be the Perron-Frobenius eigenvector for $t-H|_{\Hi_{HW}^{(n_1+\cdots + n_L - 2k-2)}}$ again in the dual canonical basis, where $t$ has been chosen so that $t-H$ is a non-negative matrix. The Perron-Frobenius theorem guarantees that in the dual canonical basis, all coefficients of $\varphi^{(n_1+\cdots+n_L - 2k)}$ are strictly positive. By the positivity property of $F$ acting on the dual canonical basis, $F\varphi^{(n_1+\cdots+n_L - 2k)}$ has strictly positive coefficients in $F\Hi_{HW}^{(n_1+\cdots+n_L-2k)}$. 
Furthermore, $F\varphi^{(n_1+\cdots + n_L-2k)}$ is an eigenvector of $H|_{F\Hi_{HW}^{(n_1+\cdots+n_L-2k)}}$ since  $H$ commutes with $\U$. And because of the decomposition given by equation (\ref{eqn:FHdecomp}), $\varphi^{(n_1+\cdots+n_L -2k-2)}$ is also a an eigenvector of $H|_{F\Hi_{HW}^{(n_1+\cdots+n_L-2k)}}$. 
By Corollary \ref{prop:posEigen}, we see that $F\varphi^{(n_1+\cdots + n_L)}$ is in fact the maximal eigenvector of $(t-H)|_{F \Hi^{(n_1+\cdots + n_L-2k)}_{HW}}$, since it also has coefficients that are strictly positive, and $H$ is a Hermitian operator. Therefore, 
\begin{equation*}
	t- E_0(H,n_1+ \cdots + n_L -2k) \geq t- E_0(H,n_1 + \cdots + n_L -2 k-2)
\end{equation*}
and the result follows. The same argument works for any $k$. 
\end{proof}

\subsection{On the Maximal Region Satisfying Sufficient Conditions for FOEL}

In the previous section, we have proven from the Perron-Frobenius Theorem that irreducibility and off-diagonal non-positivity of matrix coefficients in the dual canonical basis give sufficient conditions for FOEL. 
Without loss of generality, let us assume for the remainder of this section that $n_i\geq n_{i+1}$. 

\begin{prop}
	The maximal region of coupling coefficients which satisfy these conditions is in fact $J_k^{(i)} \leq 0$ for all $1 \leq k \leq n_{i+1}$
\end{prop}
\begin{proof}

Observe that the matrix coefficients may be computed in the following manner
\begin{equation*}
 \raisebox{-2cm}{
 \begin{tikzpicture}
 [symm/.style={draw,thick, rectangle,minimum width=1cm}]
 \node (n1) at (0,0) [symm] {};
  \node (n2) at (1.5,0) [symm] {};
   \node (n1a) at (0,3) [symm] {};
  \node (n2a) at (1.5,3) [symm] {};
   \node (n1aa) at (0,4) [symm] {};
  \node (n2aa) at (1.5,4) [symm] {};
  \node[draw,thick,rectangle,minimum width=1.5cm] (mid) at (.25,1.5) {};
  \draw (n1.90) to node[auto] {$n_i$} (mid.-140);
  \draw (n1a.-90) to node[auto,swap] {$n_i$} (mid.140);
  \draw (n2) to node[auto,swap] {$n_{i+1}-k$}(n2a);
  \draw (mid.40) to node[auto] {$k$} (n2a.-130);
  \draw (mid.-40) to node[auto] {$k$} (n2.130);
  \draw (n1a.50) to [bend left=90] node[auto] {$j$} (n2a.130);
  \draw (n1a.90) -- node[auto] {$n_i-j$} (n1aa);
  \draw (n2a.90) -- node[auto,swap] {$n_{i+1}-j$} (n2aa);
  \end{tikzpicture}
  }  
   = \sum_{\ell=0}^k Q_{jk\ell}^{(i)}
\raisebox{-1.25cm}{
\begin{tikzpicture}
 [symm/.style={draw,thick, rectangle,minimum width=1cm}]
 \node (n1) at (0,0) [symm] {};
  \node (n2) at (1.5,0) [symm] {};
  \node (n2a) at (1.5,2) [symm] {};
  \node (n1a) at (0,2) [symm] {};
\draw (n1.20) to [bend left=90] node[auto] {$j+\ell$} (n2.160);
\draw (n1a.-20) to [bend right=90] node[auto,swap] {$\ell$} (n2a.-160);
\draw (n1) to node[auto] {$n_i-j-\ell$} (n1a);
\draw (n2) to node[auto,swap] {$n_{i+1}-j-\ell$} (n2a);
\end{tikzpicture}
}
  \end{equation*}
where 
\begin{equation*}
	Q_{jk\ell}^{(i)} = 
	\begin{cases}
		\frac{[n_i]![k]![n_{i+1}-j]! [n_{i+1}-k]! [n_i +k-j-\ell]!}{[n_{i+1}]! [\ell]! [n_i+k]! [k-\ell]! [n_{i+1}-j-k]! [n_i-j-\ell]!} & \textrm{ when } \begin{cases} k \geq \ell \\ n_{i+1} -j \geq k \\ n_i -j \geq \ell \end{cases}  \\
		0 & \textrm{ otherwise } 
	\end{cases}
\end{equation*}
It should be said that the right hand side of the above equation has not been reduced to the dual canonical basis. In order to do that, one must realize the two upper projectors of size $n_i -j$ and $n_{i+1}-j$ as positive sums of elements in the Temperley-Lieb algebra. There are no loopbacks in the evaluation of this expression, and thus all terms will have the same sign as $Q_{jk\ell}^{(i)}\geq 0$. In particular, off-diagonal terms will be given by the $1 \leq \ell$ terms. Thus, sufficient conditions for FOEL are
\begin{equation*}
	\sum_{k=\ell}^{n_{i+1}} Q_{jk\ell}^{(i)} J_k^{(i)} \leq 0 \textrm{ for all } \begin{cases} 0\leq j \leq n_{i+1} \\ 1 \leq \ell \leq n_{i+1} \end{cases}
\end{equation*}

Examining the structure of the coefficients $Q_{jk\ell}^{(i)}$, we see that for a fixed $j$, $Q_{jk\ell}^{(i)}$ is a lower triangular matrix in the indices $k$ and $\ell$, and that the last $j$ rows and columns are $0$. Hence, for fixed $j$, we get the inequality 
\begin{equation*}
	Q_{jk^\prime \ell^\prime }^{(i)} J_{k^\prime}^{(i)} \leq 0 \textrm{ where $k^\prime = n_{i+1}-j$ and $\ell^\prime = n_{i+1}-j$ }
\end{equation*}
Allowing $j$ to range from $0$ to $n-1$ and noting the positivity of the non-zero $Q_{jk\ell}^{(i)}$ gives the result.
\end{proof}

By the above result, we see that negative scalar multiples of the cascade basis form extremal points of a simplex where FOEL is satisfied via the Perron-Frobenius argument. Thus we expect that if we restrict to the space $V(n_1)\otimes V(n_2)$, the cascade basis will satisfy FOEL in a highly degenerate way. Conveniently, we find that the cascade basis is already diagonal in the dual canonical basis.

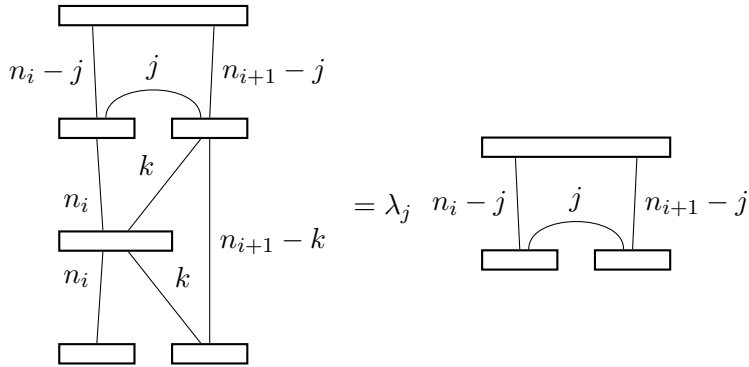
\begin{equation*}
 \raisebox{-2cm}{
 \begin{tikzpicture}
 [symm/.style={draw,thick, rectangle,minimum width=1cm}]
 \node (n1) at (0,0) [symm] {};
  \node (n2) at (1.5,0) [symm] {};
   \node (n1a) at (0,3) [symm] {};
  \node (n2a) at (1.5,3) [symm] {};
  \node (top) at (.75,4.5) [draw,thick, rectangle,minimum width=2.5cm] {};
  \node[draw,thick,rectangle,minimum width=1.5cm] (mid) at (.25,1.5) {};
  \draw (n1.90) to node[auto] {$n_i$} (mid.-140);
  \draw (n1a.-90) to node[auto,swap] {$n_i$} (mid.140);
  \draw (n2) to node[auto,swap] {$n_{i+1}-k$}(n2a);
  \draw (mid.40) to node[auto] {$k$} (n2a.-130);
  \draw (mid.-40) to node[auto] {$k$} (n2.130);
  \draw (n1a.50) to [bend left=90] node[auto] {$j$} (n2a.130);
  \draw (n1a.90) -- node[auto] {$n_i-j$} (top.-170);
  \draw (n2a.90) -- node[auto,swap] {$n_{i+1}-j$} (top.-10);
  \end{tikzpicture}
  }  = \lambda_j
  \raisebox{-.75cm}{
 \begin{tikzpicture}
 [symm/.style={draw,thick, rectangle,minimum width=1cm}]
 \node (n1) at (0,0) [symm] {};
  \node (n2) at (1.5,0) [symm] {};
  \node (top) at (.75,1.5) [draw,thick, rectangle,minimum width=2.5cm] {};
  \draw (n1.90) -- node[auto] {$n_i-j$} (top.-170);
  \draw (n2.90) -- node[auto,swap] {$n_{i+1}-j$} (top.-10);
  \draw (n1.50) to [bend left=90] node[auto] {$j$} (n2.130);
  \end{tikzpicture}
  }
\end{equation*}
where 
\begin{equation*}
	\lambda_j = \begin{cases} \frac{[n_1+k-j]![n_1]![n_2-j]![n_2-k]!}{[n_1-j]![n_1+k]![n_2-j-k]![n_2]!} & \textrm{ whenever } n_2 -j \geq k \\ 0 & \textrm{ otherwise } \end{cases}
\end{equation*}
and we have for convenience indexed the eigenvalues according to the total spin deviate, $j$, when the total spin is equal to $\frac{n_1+n_2-2j}{2}$.
Thus $\lambda_{j+1} \leq \lambda_{j}$, and for $k>1$, $K_{m,n}(k)$ has a $k$-fold degeneracy at $0$. 
\subsection*{Acknowledgments}

The work reported on in this paper was supported in part by the National Science
Foundation under Grants DMS-0757581 and DMS-1009502 (BN and SN) and 
DMS-0706927 (SS). BN gratefully acknowledges the hospitality offered by the 
Institut Mittag-Leffler (Djursholm, Sweden), where part of this work was carried out,
and support through an ESI Senior Research Fellowship from the
Erwin Schr\"odinger International Institute for Mathematical Physics, Vienna,
during the final stages of the writing this article.

\bibliographystyle{hamsplain}
\bibliography{foel}

\providecommand{\bysame}{\leavevmode\hbox to3em{\hrulefill}\thinspace}
\begin{thebibliography}{10}

\bibitem{affleck:1988}
I.~Affleck, T.~Kennedy, E.H. Lieb, and H.~Tasaki, \emph{Valence bond ground
  states in isotropic quantum antiferromagnets}, Comm. Math. Phys. \textbf{115}
  (1988), no.~3, 477--528.

\bibitem{affleck:1986}
I.~Affleck and E.H. Lieb, \emph{Proof of part of {H}aldane's conjecture on spin
  chains}, Lett. Math. Phys. \textbf{12} (1986), 57--69.

\bibitem{aizenman:1994}
M.~Aizenman and B.~Nachtergaele, \emph{Geometric aspects of quantum spin
  states}, Commun. Math. Phys. \textbf{164} (1994), 17--63.

\bibitem{AlcarazDrozHenkelRittenberg}
Francisco~C. Alcaraz, Michel Droz, Malte Henkel, and Vladimir Rittenberg,
  \emph{Reaction-diffusion processes, critical dynamics, and quantum chains},
  Ann. Physics \textbf{230} (1994), no.~2, 250--302.

\bibitem{AldousConjecture}
David Aldous, \emph{Spectral gap for the interchange (exclusion) process on a
  finite graph}, June 2009,
  \url{http://www.stat.berkeley.edu/~aldous/Research/OP/sgap.html}.

\bibitem{CantiniSportiello}
Luigi Cantini and Andrea Sportiello, \emph{Proof of the {R}azumov-{S}troganov
  conjecture}, preprint (2010), 31 pages, \url{http://arxiv.org/abs/1003.3376}.

\bibitem{Caputo}
Pietro Caputo, \emph{Energy gap estimates in {XXZ} ferromagnets and stochastic
  particle systems}, Markov Process. Related Fields \textbf{11} (2005),
  189--210.

\bibitem{CaputoLiggettRichthammer}
Pietro Caputo, Thomas~M. Liggett, and Thomas Richthammer, \emph{Proof of
  {A}ldous' spectral gap conjecture}, J.~Amer.~Math.~Soc. \textbf{23} (2010),
  831--851.

\bibitem{GoldschmidtUeltschiWindridge}
Peter~Windridge Christina~Goldschmidt, Daniel~Ueltschi, \emph{Quantum
  {H}eisenberg models and their probabilistic representations}, arXiv:1104.0983
  (2011), 48.

\bibitem{dym:2007}
Harry Dym, \emph{Linear algebra in action}, American Mathematical Society,
  2007.

\bibitem{frenkel:1997}
Igor~B. Frenkel and Mikhail~G. Khovanov, \emph{Canonical bases in tensor
  products and graphical calculus for ${U}_q (\mathfrak{sl}_2)$}, Duke Math J.
  \textbf{87} (1997), 409--480.

\bibitem{gyoja:1986}
Akihiko Gyoja, \emph{A q-analogue of {Y}oung symmetrizer}, Osaka J. Math.
  \textbf{23} (1986), 841--852.

\bibitem{hakobyan:2004}
Tigran Hakobyan, \emph{The ordering of energy levels for {SU(N)} symmetric
  antiferromagnetic chains}, Nuclear Physics B \textbf{699} (2004), 575--594.

\bibitem{hakobyan:2008}
\bysame, \emph{Antiferromagnetic ordering of energy levels for spin ladder with
  four-spin cyclic exchange: Generalization of the {L}ieb-{M}attis theorem},
  Phys. Rev. B \textbf{78} (2008), 4.

\bibitem{hakobyanLadder:2010}
\bysame, \emph{Energy-level ordering for frustrated spin ladder models},
  Physics of Atomic Nuclei \textbf{73} (2010), 339--344.

\bibitem{hakobyan:2010}
\bysame, \emph{Ordering of energy levels for extended {SU(N)} hubbard chain},
  SIGMA \textbf{6} (2010), 024.

\bibitem{jimbo:1986}
Michio Jimbo, \emph{A $q$-analogue of ${U}(\mathfrak{gl}(n+1))$, {H}ecke
  algebra, and the {Y}ang-{B}axter equation}, Lett. Math. Phys. \textbf{11}
  (1986), no.~3, 247--252.

\bibitem{kauffman:1994}
Louis~H. Kauffman and S\'ostenes~L. Lins, \emph{Temperley-{L}ieb recoupling
  theory and invariants of 3-manifolds}, no. 134, Annals of Math. Studies,
  1994.

\bibitem{kim:2003}
Dongseok Kim, \emph{Graphical calculus on representations of quantum {L}ie
  algebras}, arXiv:math/0310143v1 (2003), 63.

\bibitem{kim:2007}
\bysame, \emph{{J}ones-{W}enzl idempotents for rank 2 simple {L}ie algebras},
  Osaka J.~Math. \textbf{44} (2007), no.~3, 691--722.

\bibitem{kirillov:1989}
Anatol~N. Kirillov and Vladimir~E. Korepin, \emph{The valence bond solid in
  quasicrystals}, Algebra and Analysis \textbf{1} (1989), 47.

\bibitem{KomaNachtergaele}
Tohru Koma and Bruno Nachtergaele, \emph{The spectral gap of the ferromagnetic
  {XXZ}-chain}, Lett. Math. Phys. \textbf{40} (1997), no.~1, 1--16.

\bibitem{korepin:2006}
Vladimir~E. Korepin and Ovidiu~I. Patu, \emph{{XXX} spin chain: from {B}ethe
  solution to open problems}, Proceedings of the Solvay workshop {B}ethe
  {A}nsatz: 75 years later, 2006.

\bibitem{kuperberg:1996}
Greg Kuperberg, \emph{Another proof of the alternating sign matrix conjecture},
  Internat. Math. Res. Notices \textbf{1996(3)} (1996), 139--150.

\bibitem{kuperberg:2002}
\bysame, \emph{Symmetry classes of alternating-sign matrices under one roof},
  Ann. of Math. \textbf{156} (2002), 835--866.

\bibitem{lieb:1962}
Elliot~H. Lieb and Daniel Mattis, \emph{Ordering energy levels of interacting
  spin systems}, J. Math. Phys. \textbf{3} (1962), 749--751.

\bibitem{masbaum:1994}
G.~Masbaum and P.~Vogel, \emph{3-valent graphs and the {K}auffman bracket},
  Pacific J. Math. \textbf{164} (1994), 361--381.

\bibitem{nachtergaele:1994}
Bruno Nachtergaele, \emph{A stochastic geometric approach to quantum spin
  systems}, Probability and Phase Transition, NATO SCIENCE SERIES: C:
  Mathematical and Physical Sciences (Geoffrey Grimmett, ed.), vol. 420,
  Kluwer, 1994, pp.~237--246.

\bibitem{nachtergaele:2004}
Bruno Nachtergaele, Wolfgang Spitzer, and Shannon Starr, \emph{Ferromagnetic
  ordering of energy levels}, J. of Stat. Phys. \textbf{116} (2004), 719--738.

\bibitem{nachtergaele:2005}
Bruno Nachtergaele and Shannon Starr, \emph{Ferromagnetic {L}ieb-{M}attis
  theorem}, Phys. Rev. Lett. \textbf{94} (2005), 057206.

\bibitem{nachtergaele:2006}
\bysame, \emph{Ordering of energy levels in {H}eisenberg models and
  applications}, Lecture Notes in Physics \textbf{690} (2006), 149--170.

\bibitem{Thomas}
Lawrence~E. Thomas, \emph{Quantum {H}eisenberg ferromagnets and stochastic
  exclusion processes.}, J.~Math.~Phys. \textbf{21} (1980), no.~7, 1921--1924.

\bibitem{westbury:1995}
B.~Westbury, \emph{The representation theory of the {T}emperley-{L}ieb
  algebras}, Mathematische Zeitschrift \textbf{219} (1995), 539--565,
  10.1007/BF02572380.

\end{thebibliography}

\end{document}